\documentclass[english,11pt,onecolumn,draftclsnofoot]{IEEEtran}
\usepackage[T1]{fontenc}
\usepackage[latin9]{inputenc}
\usepackage{color}
\usepackage{amsmath}
\usepackage{amsthm}
\usepackage{amssymb}
\usepackage{graphicx}

\makeatletter

\providecommand{\tabularnewline}{\\}

\theoremstyle{plain}
\newtheorem{thm}{\protect\theoremname}
\theoremstyle{plain}
\newtheorem{cor}{\protect\corollaryname}
\theoremstyle{plain}
\newtheorem{lem}{\protect\lemmaname}

\usepackage{amsfonts}
\usepackage{cite}
\usepackage{array}
\usepackage{algorithm}
\usepackage{algorithmic}
\usepackage{subfigure}
\usepackage{stfloats}

\makeatother

\usepackage{babel}

\makeatother

\usepackage{babel}
\providecommand{\corollaryname}{Corollary}
\providecommand{\lemmaname}{Lemma}
\providecommand{\theoremname}{Theorem}

\begin{document}

\title{Practical Modeling and Analysis of \\
Blockchain Radio Access Network}

\author{Xintong~Ling,~\IEEEmembership{Member,~IEEE}, Yuwei Le,~\IEEEmembership{Student
Member,~IEEE},\\
 Jiaheng~Wang,~\IEEEmembership{Senior Member,~IEEE}, Zhi~Ding,~\IEEEmembership{Fellow,~IEEE},
Xiqi~Gao,~\IEEEmembership{Fellow,~IEEE}\thanks{X. Ling, Y. Le, J. Wang, and X. Gao are with the National Mobile Communications
Research Laboratory, Southeast University, Nanjing, China (e-mail:
xtling@seu.edu.cn, ywle@seu.edu.cn, jhwang@seu.edu.cn, and xqgao@seu.edu.cn).
Z. Ding is with Department of Electrical and Computer Engineering,
University of California, Davis, California, 95616 (e-mail: zding@ucdavis.edu).}\vspace{-0.2cm}}
\maketitle
\begin{abstract}
The continually rising demand for wireless services and applications
in the era of Internet of things (IoT) and artificial intelligence
(AI) presents a significant number of unprecedented challenges to
existing network structures. To meet the rapid growth need of mobile
data services, blockchain radio access network (B-RAN) has emerged
as a decentralized, trustworthy radio access paradigm spurred by blockchain
technologies. However, many characteristics of B-RAN remain unclear
and hard to characterize. In this study, we develop an analytical
framework to model B-RAN and provide some basic fundamental analysis.
Starting from block generation, we establish a queuing model based
on a time-homogeneous Markov chain. From the queuing model, we evaluate
the performance of B-RAN with respect to latency and security considerations.
By connecting latency and security, we uncover a more comprehensive
picture of the achievable performance of B-RAN. Further, we present
experimental results via an innovative prototype and validate the
proposed model.
\end{abstract}

\section{Introduction}

The past decade has witnessed tremendous growth in emerging wireless
technologies geared toward diverse applications\cite{Cisco2019}.
Radio access networks (RANs) are becoming more heterogeneous and highly
complex. Without well-designed inter-operation, mobile network operators
(MNOs) must rely on their independent infrastructures and spectra
to deliver data, often leading to duplication, redundancy, and inefficiency.
A huge number of currently deployed business or individual access
points (APs) have not been coordinated in the existing architecture
of RANs, and are therefore under-utilized. Meanwhile, user equipments
(UEs) are not granted to access to APs of operators other than their
own, even though some of them may provide better link quality and
economically sensible. The present state of rising traffic demands
coupled with the under-utilization of existing spectra and infrastructure
resources motivates the development of a novel network architecture
to integrate multiple parties of service providers (SPs) and clients
to transform the rigid network access paradigm that we face today.

Recently, blockchain has been recognized as a disruptive innovation
shockwave \cite{Tschorsch2016,Xie2019,Dai2019}. Federal Communications
Commission (FCC) has been suggested that blockchain may be integrated
into wireless communications for the next-generation network (NGN)
in the Mobile World Congress 2018. Along the same line, the new concept
of blockchain radio access network (B-RAN) was formally proposed and
defined in \cite{Ling2019}. In a nutshell, B-RAN is a decentralized
and secure wireless access paradigm that leverages the principle of
blockchain to multiple trustless networks into a larger shared network
and benefits multiple parties from positive network effects. As revealed
in \cite{Ling2019}, B-RAN can drastically improve network throughput
via cross-network sharing and offloading. Furthermore, the positive
network effect can help B-RAN recruit and attract more players, including
network operators, spectral owners, infrastructure manufacturers,
and service clients alike. The subsequent expansion of such a shared
network platform would make the network platform more valuable, thereby
generating a positive feedback loop. In time, a vast number of individual
APs can be organized into B-RAN and commodified to form a sizable
and ubiquitous wireless network, which can significantly improve the
utility of spectra and infrastructures. In practice, rights, responsibilities
and obligations of each participant in B-RAN can be flexibly codified
as smart contracts executed by blockchain.

Among existing studies on leveraging blockchain in networks, most
have focused on Internet of things (IoT)\cite{Christidis2016,Novo2018,Danzi2018,Ling2020,Backman2017},
cloud/edge computing\cite{Xiong2019,Yang2019,Liu2018}, wireless sensor
networks \cite{Yang2019a}, and other consensus mechanisms\cite{Liu2019,Wang2019}.
Only a few considered the future integration of blockchain in wireless
communications. Weiss \emph{et al.} \cite{Weiss2019} discussed several
potentials of blockchain in spectrum management. Kuo \emph{et al.}
\cite{Kuo2018}\emph{ }summarized some critical issues when applying
blockchain to wireless networks and pointed out the versatility of
blockchain. Pascale \emph{et al. }\cite{Pascale2017} adopted smart
contracts as an enabler to achieve service level agreement (SLA) for
access. Kotobi \emph{et al.} \cite{Kotobi2018} proposed a secure
blockchain verification protocol associated with virtual currency
to enable spectrum sharing.  Le \emph{et al. }\cite{Le2019} developed
an early prototype to demonstrate the functionality of B-RAN.

Despite the growing number of papers and heightened interests with
respect to blockchain, works including fundamental analysis are rather
limited. A number of critical difficulties remain unsolved. 1) Existing
works have not assessed the impact of decentralization on RANs after
introducing blockchain. Decentralization always comes with a price
that should be characterized and quantified. 2) Very few papers have
noticed that service latency will be a crucial debacle for B-RAN as
a price of decentralization\cite{Pascale2017}. Unfortunately, the
length of such delay and its controllability are still open issues.
3) Security is yet another critical aspect of blockchain-based protocols.
In particular, alternative history attack, as an inherent risk of
decentralized databases, is always possible, and must be assessed.
4) A proper model is urgently needed to exploit the characteristics
of B-RAN (such as latency and security) and to further provide insights
and guidelines for real-world implementations.

\begin{table}
\caption{Important variables in the modeling and analysis\label{table1}}
\begin{raggedright} \centering%
\begin{tabular}{|c|c|c|c|}
\hline
Symbols & Explanations & Symbols & Explanations\tabularnewline
\hline
\hline
$t_{k}^{a}$ & Arrival epoch of request $k$ & $\tau_{k}^{c}$ & Required service time of request $k$\tabularnewline
\hline
$\lambda^{a}$ & Requests arrival rate & $T^{a}=1/\lambda^{a}$ & Average inter-arrival time\tabularnewline
\hline
$\lambda^{b}$ & Blocks generation rate & $T^{b}=1/\lambda^{b}$ & Average block time\tabularnewline
\hline
$\lambda^{c}$ & Service rate & $T^{c}=1/\lambda^{c}$ & Average service time\tabularnewline
\hline
$s$ & Number of access links in the tract & $N$ & Number of required confirmations\tabularnewline
\hline
$\rho=\frac{\lambda^{a}}{s\lambda^{c}}$ & Traffic intensity & $\beta$ & Relative mining rate of an attacker\tabularnewline
\hline
$\Phi=\left\{ \lambda^{a},\lambda^{b},\lambda^{c},s\right\} $  & Basic configuration of B-RAN &  & \tabularnewline
\hline
\end{tabular}\end{raggedright} \centering{}\vspace{-0.6cm}
\end{table}
To address the aforementioned open issues, this study aims to establish
a framework to concretely model and evaluate B-RAN. We start from
the block generation process and develop an analytical model to characterize
B-RAN behaviors. We shall evaluate the performance in terms of latency
and security in order to present a more comprehensive view of B-RAN.
We will verify the efficacy of our model by building an innovative
B-RAN prototype. We summarize our key contributions as follows:
\begin{itemize}
\item We clearly define the workflow of B-RAN and develop an original queuing
model based on a time-homogeneous Markov chain. To the best of our
knowledge, this is the first known analytical model for B-RAN.
\item From the queuing model, we analytically characterize the system latency
of B-RAN and further derive both upper and lower bounds on B-RAN latency.
\item We use the probability of successful attack to define the security
level of B-RAN and evaluate potential factors that influence the security.
We further assess the risk by taking the attacker's strategy into
consideration.
\item Based on the modeling and analysis, we uncover an inherent trade-off
relationship between security and latency, and develop an in-depth
understanding regarding the achievable performance of B-RAN.
\item Finally, we build a B-RAN prototype that can be used comprehensive
experiments to validate the accuracy of our analytical model and results.
\end{itemize}
We organize this manuscript as follows. Section II presents the B-RAN
framework and the prototype. Section III provides the mining model
to describe the block generation process. In Section IV, we establish
the B-RAN queuing model, with which we analyze and evaluate the B-RAN
performance concerning latency and security in Section V and Section
VI, respectively. We demonstrate the latency-security trade-off in
Section VII and provide some in-depth insights to B-RAN. Section VIII
concludes our manuscript. Given the large number of symbols to be
used, we summarize the important variables in Table \ref{table1}.

\section{B-RAN Framework}

\begin{figure*}
\begin{raggedright} \centering\includegraphics[width=0.8\textwidth]{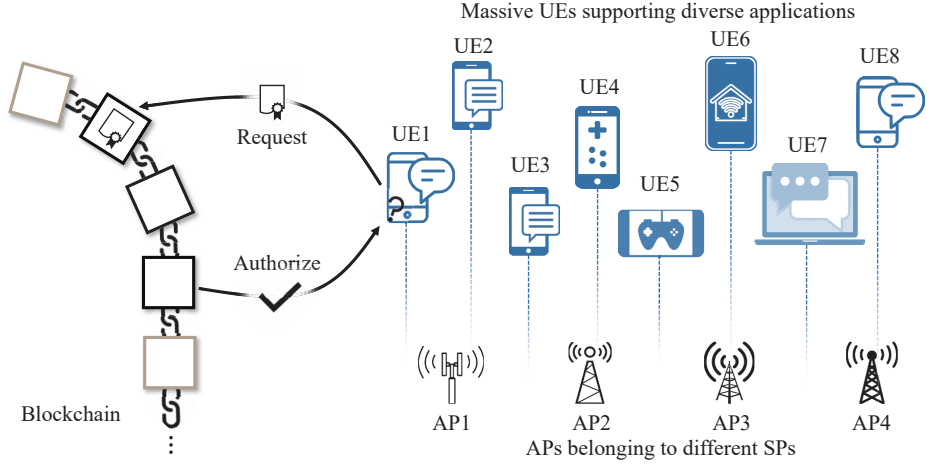}

\end{raggedright} \centering{}

\caption{Conceptual illustration of self-organized B-RAN.\label{fig:illustration}}
\vspace{-0.6cm}
\end{figure*}

B-RAN can integrate multiple networks across SPs for diverse applications.
As illustrated in Fig. \ref{fig:illustration}, B-RAN is self-organized
by APs belonging to multiple SPs, massive UEs, and a blockchain maintained
by miners. In B-RAN, there is a confederacy of SPs (organizations
or individuals) that are willing to provide public wireless access
under shared control. These SPs in B-RAN allow the greater pool of
UEs to access their APs and networks by receiving payment or credit
for reciprocal services. Blockchain, as the name suggests, is a chain
of interconnected blocks\cite{Nakamoto2008}, which acts as a public
ledger in B-RAN for recording, confirming, and enforcing digital actions
in smart contracts.

Traditional roaming in cellular networks may be viewed as a legacy
solution that has poor cost structure. There has been little incentive
for cellular RAN operators to coordinate services and coverages. It
is so challenging to reach mutual agreement on roaming between hundreds
of organizational SPs without relying on a trusted third party, not
to include unknown individual SPs. During roaming, UEs are also required
to select an AP from its subscribed operator, even though it may not
deliver services of the most bandwidth-efficient and highest quality
in a given environment. With the help of blockchain, B-RAN can form
an expansive cooperative network of different operators to deliver
high quality service at high spectrum efficiency while protecting
the interests of all legal participants at the same time.

\subsection{Access Workflow}

Since B-RAN is envisioned to be broadly inter-operative and to support
multiple advanced wireless services and standards. This work focuses
on the most basic access approach for which the procedure is shown
in Fig. \ref{fig:workflow}.
\begin{itemize}
\item In preparation for access, UEs and SPs should first enter an SLA containing
the details including service types, compensation rates, among other
terms. (For example, SPs can first publish their service quality and
charge standard, and UEs select suitable SPs according to the expenditure
and quality of service.) The service terms and fees will be explicitly
recorded in a smart contract authorized by the digital signatures
of both sides.
\item In step 1, the smart contract with the access request is committed
to the mining network and is then verified by miners.
\item In step 2, the verified contracts are assembled into a new block,
which is then added at the end of the chain.
\item In step 3, the block is accepted into the main chain after sufficient
blocks as confirmations built on top of it.
\item In step 4, the access service is delivered according to the signed
smart contract.
\end{itemize}
Clients can obtain access services more conveniently through the above
process instead of signing contracts with a specific MNO in advance.
The service duration in B-RAN is flexible and can be as short as a
few minutes or hours, which is different from typical long-term plans
(e.g., monthly plans). Hence, in B-RAN, mobile devices can access
suitable APs belonging to various SPs which likely provide higher
quality coverage for the UEs in their current locations. UEs can prolong
access services by renewing the contract earlier before the previous
one expires in order to continue the connection status. Hence, service
latency in this context refers to the delay when a UE accesses an
unknown network for the first time, rather than the delay in the physical
layer transmission.

Mathematically, we can describe the request structure by $\mathsf{REQ}_{k}\left(t_{k}^{a},\tau_{k}^{c}\right)$
shown in Fig. \ref{fig:workflow}, where $t_{k}^{a}$ and $\tau_{k}^{c}$
are the arrival epoch and the service duration of request $k$, respectively.
Assume that the access requests are mutually independent, and arrive
as a Poisson process with rate $\lambda^{a}$. Equivalently, the inter-arrival
time between two requests $U^{a}$ follows exponential distribution
with mean $T^{a}=1/\lambda^{a}$. Based on well-known studies such
as \cite{Cooper1972}, the random service time $\tau_{k}^{c}$ is
also expected to be exponential with mean $T^{c}=1/\lambda^{c}$.
Note that in this work, we consider a tract covered by multiple trustless
SPs (organizations or individuals). The APs belonging to these SPs
in the tract are capable of providing access services for at most
$s$ UEs. In other words, $s$ can be viewed as the maximum number
of access links in the considered tract. The block size limit is considered
as unlimited when we focus on one tract.

\subsection{Consensus Mechanism}

\begin{figure*}
\begin{raggedright} \centering\includegraphics[width=0.8\textwidth]{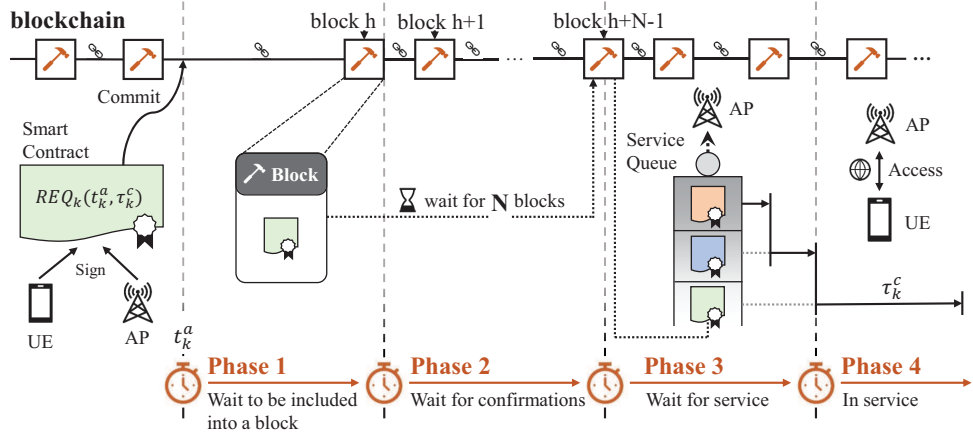}

\end{raggedright} \centering{}\caption{Four stages of the access workflow in B-RAN.\label{fig:workflow}}
\vspace{-0.6cm}
\end{figure*}
B-RAN, as a decentralized system, requires proper consensus mechanisms
for consistency\cite{Xu2019,Wang2019}. Proof-of-Work (PoW) has been
widely used in practice and proven to be secure in cryptocurrencies
such as Bitcoin. In PoW, network maintainers, also known as miners,
need to obtain a hash value below a given target by repeatedly guessing
a random variable named nonce. However, PoW-based consensus mechanisms
consume a tremendous amount of energy, which is likely unbearable
for energy-limited mobile devices. Consequently, Proof-of-Device (PoD)
is proposed for B-RAN as a promising alternative by utilizing the
fact that wireless access usually depends on a hardware device associated
with a unique identifier (ID)\cite{Ling2019}. PoD deploys unique
tamper-proof IDs, such as the international mobile equipment identity
(IMEI) and Identifier for Advertisers (IFA), in order to distinguish
different network entities. Also, due to variations during manufacture,
every device has multiple hardware-dependent features, which could
constitute a unique RF fingerprinting for each device and can be identified
from the transmitted RF signal\cite{Wu2018}. In the real world, forging
an identity of a device is often costly, whereas creating multiple
identities is almost costless in cryptocurrencies. Therefore, PoD
can safeguard the security of B-RAN without expending immense computing
power and is thus suitable for wireless networks. Notably, PoW, PoD,
and other alternatives can be put in the same class since all of them
are based on hash puzzles. We will further discuss and model the block
generation process of a hash-based consensus mechanism in Section
III.

\subsection{Implementation}

\textcolor{blue}{}In order to evaluate our established model, we
will provide demonstrative experimental results from a home-built
prototype throughout the whole article. We implement this version
of B-RAN prototype on a workstation with Intel Core CPU I7-8700K and
32GB RAM. Our prototype consists of a standard file system for data
storage, a key-value database for file index, and several core modules
written in Python. The prototype supports both PoW and PoD consensus
mechanisms and adopts the fast smart contract deployment (FSCD)\cite{Le2019}.
We configure different Ethernet ports as UEs and APs and set up the
integrated development environment (IDE), wherein UEs propose random
access requests according to the input configurations and APs provide
services based on the workflow given in Section II-A. During tests,
the prototype can track running statistics and provide them as output
results. More details about the implementation, e.g., the concepts
of root contract and contract cache, can be found in \cite{Le2019}.
Note that, in this work, the average service time $T^{c}$ as time
unit is set to unity without loss of generality, so time is measured
as relative variables in terms of time unit $T^{c}$.

\section{Mining Model}

\subsection{Hash-Based Mining}

In this section, we will present a general model for hash-based
puzzles to describe the block generation process, also known as mining.
Usually, block propagation in networks is much faster than block generation,
and we thus ignore the block spreading delay. We will verify the mining
model by using the real data of Bitcoin and Ethereum.

Generally, a hash puzzle is to find a suitable answer to satisfy the
following conditions:
\begin{equation}
\text{Hash}(\text{HP}+\text{DP}+\text{TS}+\text{OF})<\text{GT}.\label{mp:HashPuzzle}
\end{equation}
Here \textquotedblleft HP\textquotedblright{} stands for the hash
pointer to a previous block, ``DP'' means the data payload, ``TS''
is the current timestamp, ``GT'' represents a given target, and
``OF'' is the optional field depending on the specific type of the
hash puzzle. For instance, in the PoW protocol, the optional field
can be any random number, whereas, in the PoD protocol, the optional
field is given by the hardware ID. In PoW, the range of the optional
field is unlimited. Hence, a miner can guess many times to find a
correct nonce, and hence the number of trials is only restricted by
the mining rigs. In PoD, the optional field is given by the tamper-proof
ID such that each device can perform the hash computation only once
for each slot, thereby largely reducing the power consumption. The
premise behind that is the entities in real RANs cannot be effortlessly
forged or created. The characteristics, e.g., security level and power
consumption, of different hash-based consensus mechanism can be traded
off by properly choosing the optional field.

\begin{table}
\caption{Notations in the model of mining.}
\begin{raggedright} \centering%
\begin{tabular}{|c|c|c|c|}
\hline
Symbols & Explanations & Symbols & Explanations\tabularnewline
\hline
\hline
$p$ & Success probability of a single hash trial &  & \tabularnewline
\hline
$m$ & Number of hash trials & $\tau$ & Length of time conducted $m$ trials\tabularnewline
\hline
$W^{b}$ & Number of failures preceding the first success & $U^{b}$ & Continuous time preceding the first success\tabularnewline
\hline
$\lambda^{b}=\frac{mp}{\tau}$ & Mean number of successes per unit time (success rate) & $T^{b}=\frac{1}{\lambda^{b}}$ & Average block time\tabularnewline
\hline
$B(n,t,t+h)$ & Event that $n$ blocks occur in the interval $(t,t+h)$ & $\beta$ & Relative mining rate of an attacker\tabularnewline
\hline
\end{tabular}\end{raggedright} \centering{}

\vspace{-0.6cm}
\end{table}

\subsection{Modeling of Hash Trials}

For a general hash-based mining process, each hash trial can be regarded
as an independent Bernoulli experiment with success probability $p$
as the timestamp keeps changing. In a sequence of independent Bernoulli
trials, the probability that the first block is generated after exactly
$m$ failures is $\left(1-p\right)^{m}p$. Let $W^{b}$ be the number
of failures preceding the first success. Then, $W^{b}$ follows geometric
distribution:
\begin{equation}
\text{Pr}\left\{ W^{b}=m\right\} =\left(1-p\right)^{m}p,\thinspace\thinspace m=0,1,...\label{mp:wbpmf}
\end{equation}
Also, the probability that the number of trials preceding the first
success is at least $m$ equals
\begin{equation}
\text{Pr}\left\{ W^{b}\geq m\right\} =\left(1-p\right)^{m},\thinspace\thinspace m=0,1,...\label{mp:wbcdf}
\end{equation}
Here, $W^{b}$ can be viewed as the waiting period before a block
is successfully generated, and its distribution can be described by
\eqref{mp:wbpmf} and \eqref{mp:wbcdf}. The average number of successes
in $m$ independent trials is $mp$.

Hence, if $m$ hash trials are conducted in a time interval of length
$\tau$, then the success rate defined as the mean number of successes
per unit time is $\lambda^{b}=mp/\tau$. Now let $p\rightarrow0$
and $m\rightarrow\infty$ in the way that keeps $\lambda^{b}$ constant.
We can visualize an experiment with infinite hash trials performed
within interval $\tau$. Successive trials are infinitesimally close
with vanishingly small probability of success, but the mean number
of successes remains a non-zero constant $\lambda^{b}\tau$. By using
the fact that the geometric distribution approaches the exponential
distribution in the limit, we have
\begin{equation}
\lim_{m\rightarrow\infty}(1-p)^{m}\mid_{p=\frac{\lambda^{b}\tau}{m}}=\lim_{m\rightarrow\infty}\left(1-\frac{\lambda^{b}\tau}{m}\right)^{m}=\exp(-\lambda^{b}\tau).\label{eq:limitcondition}
\end{equation}
Define a random variable $U^{b}$ as the continuous time before preceding
the first success. Since $p=\frac{\lambda^{b}\tau}{m}$, then we have
\begin{equation}
\text{Pr}\left\{ U^{b}>\tau\right\} =\text{Pr}\left\{ W^{b}>m\right\} =\exp(-\lambda^{b}\tau).\label{eq:expdis}
\end{equation}
$U^{b}$ is also named as the block time in the context of blockchain.
The average block time, denoted by $T^{b}$, is equal to $1/\lambda^{b}$
according to \eqref{eq:expdis}. $\lambda^{b}$ is thus called as
the mining rate representing the block generation rate.

\eqref{eq:limitcondition} and \eqref{eq:expdis} imply that, if the
number of hash computations of the whole network in a unit time tends
to infinity, then the length of time between two successive blocks,
i.e., $U^{b}$'s, would follow the exponential distribution. Interestingly,
most of mature PoW blockchain networks indeed perform a huge number
of hash computations every moment. For example, the minimum hash rates
of bitcoin and Ethereum during 2018 are 14,891 TH/s and 159TH/s\footnote{Source: bitinfocharts.com/, accessed Nov., 2019.}.
These are humongous numbers in real world, and practically support
the limiting condition that the number of trials tends to infinity.
When there are massive hardware devices participating in mining, the
condition also holds for PoD. Note that the block times $U^{b}$'s
are always mutually independent and identically distributed because
of the memoryless property of geometric distribution, even if the
exponential approximation no longer holds.

    Now we further prove that blocks generate form a Poisson
process. Let $B(n,t,t+h)$ denote the event that $n$ blocks are generated
in interval $(t,t+h)$. As $h\rightarrow0$, the probability that
at least one block will is generated in $(t,t+h)$ is
\begin{align*}
\text{Pr}\left\{ B(n\geq1,t,t+h)\right\}  & =1-\text{Pr}\left\{ B(n=0,t,t+h)\right\} =1-\text{Pr}\left\{ U^{b}>h\right\} \\
 & =1-\exp(-\lambda^{b}h)=\lambda^{b}h+o(h),\thinspace\,h\rightarrow0.
\end{align*}
We note that $o(h)$ is an infinitesimal of higher order such that
$\lim_{h\rightarrow0}o(h)/h=0$. The probability that exactly one
block will occur in $(t,t+h)$ is
\begin{align*}
\text{Pr}\left\{ B(n=1,t,t+h)\right\}  & =\int_{0}^{h}\text{Pr}\left\{ U_{1}^{b}=h-u\right\} \text{Pr}\left\{ U_{2}^{b}>u\right\} du\\
 & =\int_{0}^{h}\lambda^{b}\exp(-\lambda^{b}(h-u))\exp(-\lambda^{b}u)du\\
 & =\lambda^{b}h\exp(-\lambda^{b}h)=\lambda^{b}h+o(h),\thinspace\thinspace h\rightarrow0,
\end{align*}
where $U_{1}^{b}$ and $U_{2}^{b}$ represent the generation times
of first and second blocks since the epoch $t$. As $h\rightarrow0$,
the probability of occurrence of two or more blocks in $(t,t+h)$
is
\begin{align*}
\text{Pr}\left\{ B(n\geq2,t,t+h)\right\} =1-\exp(-\lambda^{b}h)-\lambda^{b}h\exp(-\lambda^{b}h) & =o(h),\thinspace\thinspace h\rightarrow0.
\end{align*}
 In conclusion, the block generation can be modeled as a Poisson process
with mining rate $\lambda^{b}$. The notations in this section are
summarized in Table II.

In\textbf{ Prototype Verification A}, we illustrate the model validity
of Poisson in Fig. \ref{fig:realdata} by empirical data. In our self-built
B-RAN prototype described in Fig. \ref{fig:realdata}(a), we deploy
five miners with equal mining rates and set the block generation rate
to $\lambda^{b}=1/10$, or equivalently the average block time $T^{b}=1/\lambda^{b}=10$.
We measure the block time for 10,000 blocks and plot the histogram
to approximate the distribution of block time $U^{b}$. One can see
that the Poisson model closely fits the data. Furthermore, we collected
the empirical data from Bitcoin (51,036 blocks starting at height
530,114\footnote{Source: www.blockchain.com/en/btc/blocks/, accessed Nov., 2019.})
and Ethereum (2,078,000 blocks starting at height 5,806,289\footnote{Source: etherscan.io/block/, accessed Nov., 2019.}),
respectively, to demonstrate the validity of the model in practice.
The collected Bitcoin and Ethereum data are also consistent with the
model in a real network where propagation delay does exist. Furthermore,
the measured average block times of both Bitcoin and Ethereum are
close to the block times predefined by protocols. Therefore, even
though the propagation delay is simply neglected, the Poisson model
can still well characterize the block generation process in practice.

\begin{figure*}
\begin{raggedright} \centering \subfigure[]{\includegraphics[width=0.3\textwidth]{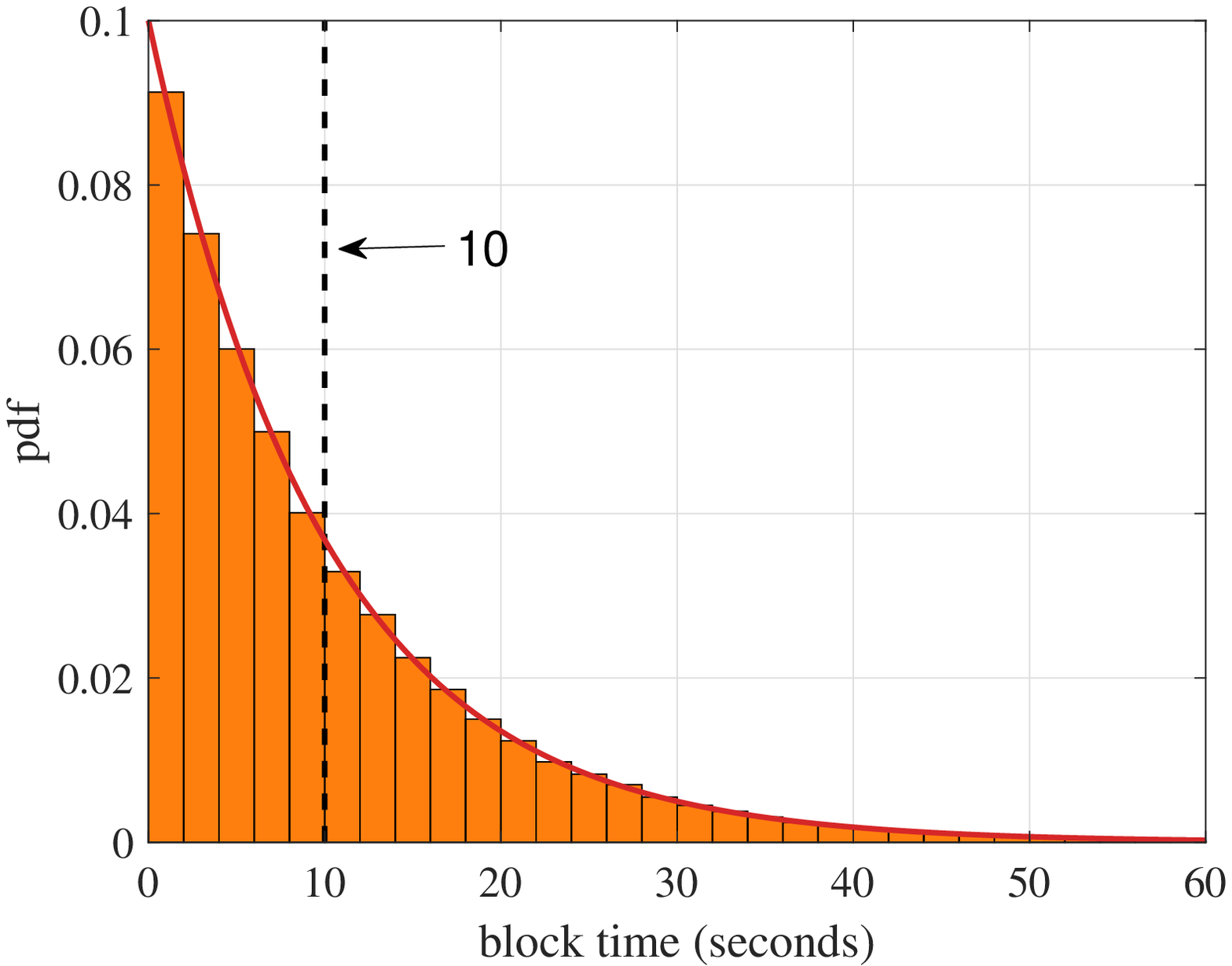}}
\hfill{}\subfigure[]{\includegraphics[width=0.3\textwidth]{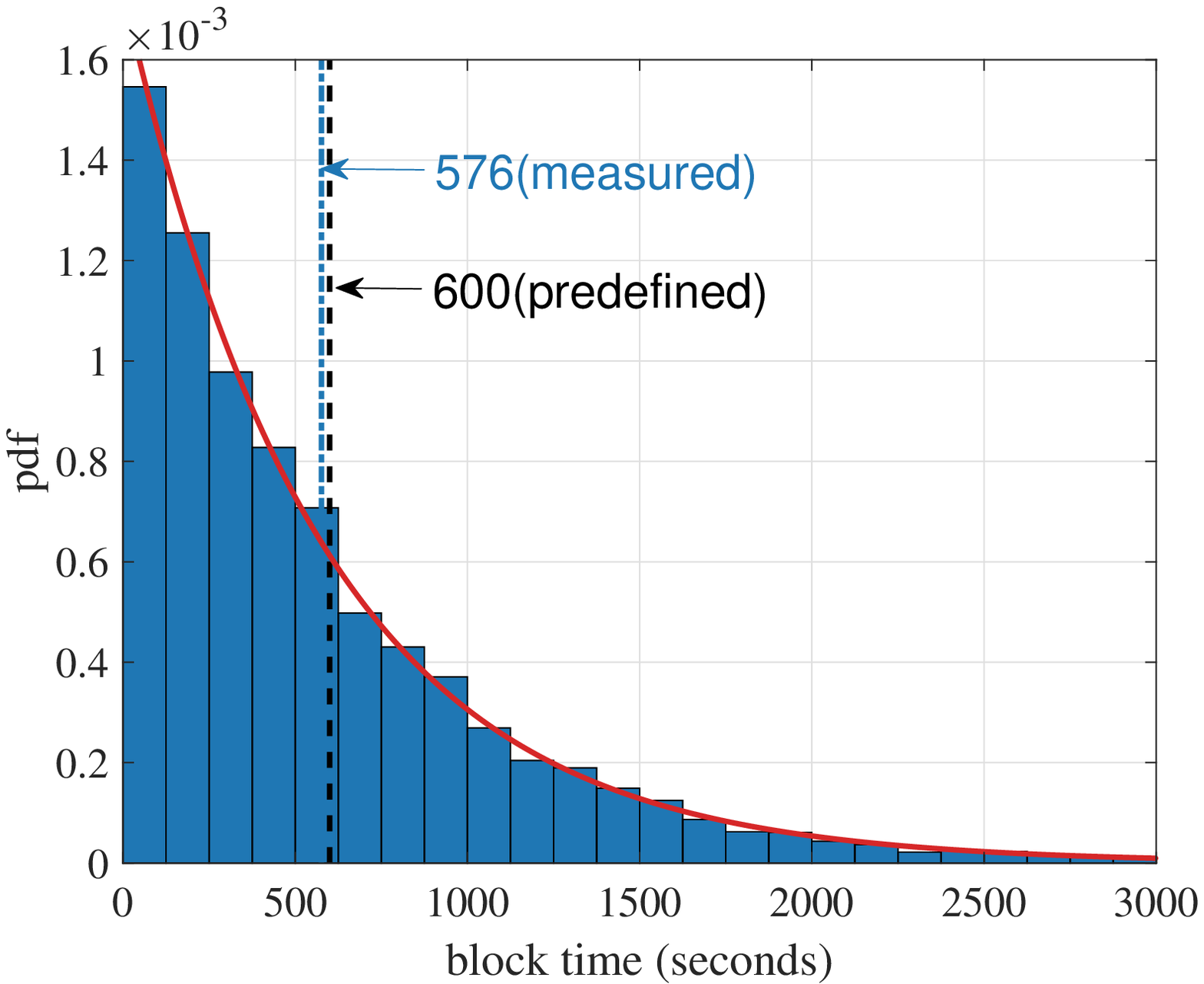}}
\hfill{}\subfigure[]{\includegraphics[width=0.3\textwidth]{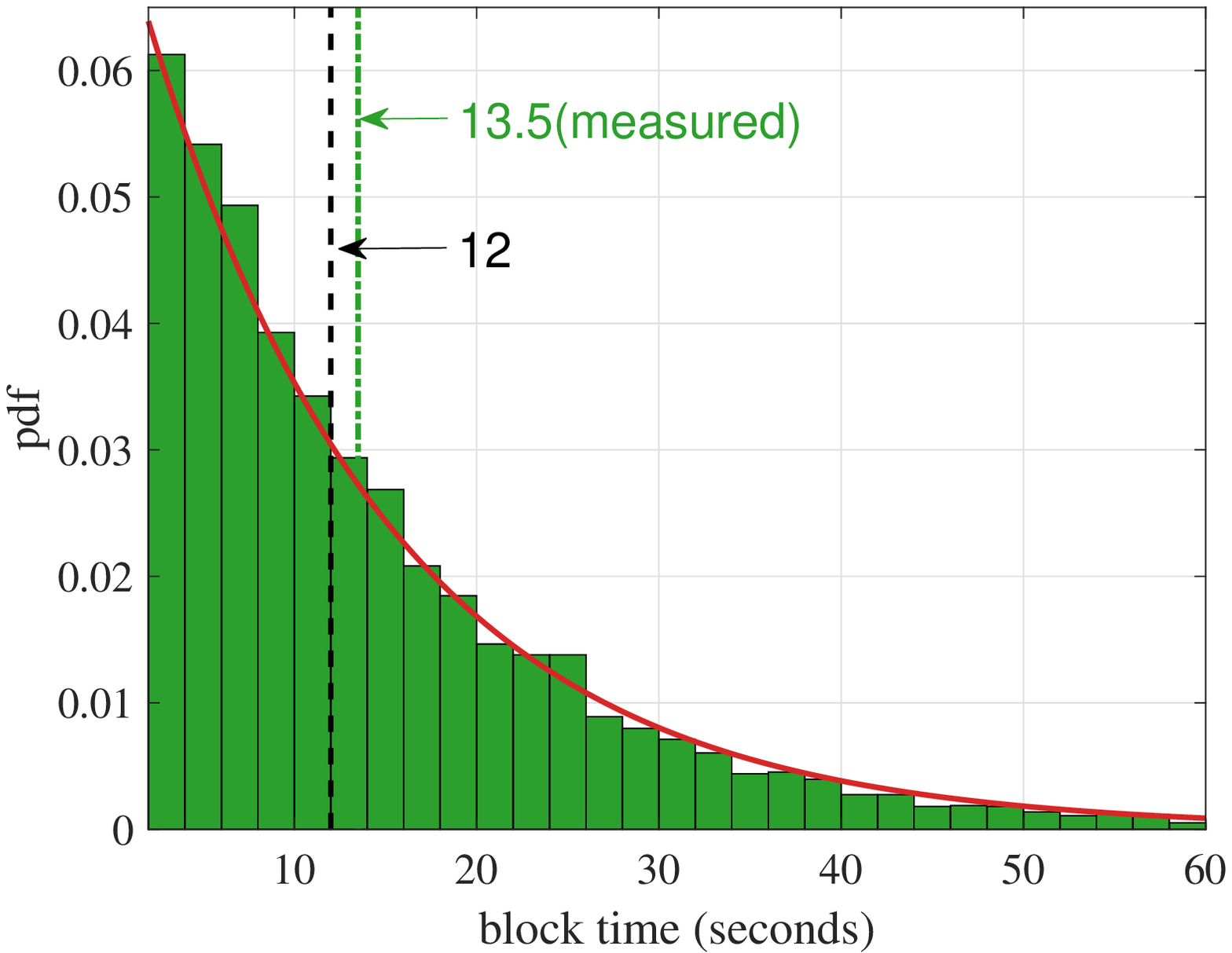}}

\end{raggedright} \centering{}\caption{The distribution of block time from real data and simulations. (Please
see footnotes 2 and 3 for sources.) (a) Self-built B-RAN prototype.
(b) Bitcoin. (c) Ethereum. \label{fig:realdata}}
\vspace{-0.6cm}
\end{figure*}

\subsection{Attacker's Model}

In this part, we consider an adversary who mounts an alternative
history attack by attempting to generate a longer fraudulent block
chain. The adversarial has to generate new blocks in the same way
as honest miners since the hash value can hardly be tampered. Hence,
the fraudulent blocks are also generated as Poisson. Now assume that
the mining rate of the attacker is $\beta\lambda^{b}$, while the
mining rate of honest miners is $\lambda^{b}$, i.e., the attacker
controls $\frac{\beta}{1+\beta}$ fraction of hash power\footnote{Usually, an attacker can hardly amass more hash power than the sum
of other honest miners, i.e., $\beta<1$; Otherwise, the attacker
already dominates the mining network.}. Both the attacker and the benign network are mining independently.
From the additive property of Poisson processes, the sum generation
rate of both benign and fraudulent block is $\left(\beta+1\right)\lambda^{b}$.
Note that, in an alternative history attack, the attacker will not
publish the fraudulent chain until it creates a more extended branch.
Consequently, the benign participants of B-RAN are unaware of the
existence of the attacker, and can only observe the blocks generated
by honest miners with the generation rate $\lambda^{b}$ until a fraud
succeeds. We will provide a detailed procedure of the alternative
history attack in Section VII.

\section{B-RAN Queuing Model}

 According to the B-RAN framework in Section II, we divide the service
process of a valid request into four stages: 1) waiting to be included
into a block; 2) waiting for confirmations; 3) waiting for service;
4) in service.  Naturally, we can model the process using several
queues in tandem based on the four phases. However, it is worth pointing
out that, in the third stage, the requests in the same block arrive
simultaneously, and the number of requests is related to the block
generation time $U^{b}$.  Hence, this queue is non-Markovian since
some previous events (e.g., $U^{b}$) beside the current state of
the queue may affect its future state. Usually, a non-Markovian queue
is difficult to tackle. Therefore, we should carefully select the
state space of B-RAN for further analysis.

Let $i_{n}$ be the number of pending requests that already have $n$
confirmations. A pending request is confirmed after received $N$
confirmations. Then $j=\sum_{n=N}^{+\infty}i_{n}$ denotes as the
number of confirmed requests that have not been served yet. In this
approach, B-RAN can be fully identified by state $E(i_{0},i_{1},...,i_{N-1},j)$
belonging to the $\left(N+1\right)$-dimensional state space $\mathbb{Z}_{+}^{N+1}$,
where $\mathbb{Z}_{+}^{N+1}$ represents the set of all $\left(N+1\right)$-tuples
of non-negative integers.  Formally, the queuing process of B-RAN
can be completely described by a vector stochastic process as follows:
\[
\{X(t)\in\mathbb{Z}_{+}^{N+1},t\geq0\}.
\]
 B-RAN is said to be in state $E(i_{0},i_{1},...,i_{N-1},j)$ at
time $t$ if $X(t)=E(i_{0},i_{1},...,i_{N-1},j)$. Note that the way
to establish a queuing model is not unique. We define B-RAN by using
the queuing model $\{X(t),t\geq0\}$ owing to two critical properties,
as shown in Theorem \ref{thm:queuemodel}.
\begin{thm}
\label{thm:queuemodel}The queuing model $\left\{ X(t),\thinspace t\ge0\right\} $
is a continuous time-homogeneous Markov process with properties:
\begin{enumerate}
\item Markov:
\[
\text{Pr}\left\{ X(t+h)=E|X(t)=E^{\prime},X(u)\text{ for }0\le u\le t\right\} =\text{Pr}\left\{ X(t+h)=E|X(t)=E^{\prime}\right\} ;
\]

\item Time homogeneity:
\[
\text{Pr}\left\{ X(t+h)=E|X(t)=E^{\prime}\right\} =\text{Pr}\left\{ X(t)=E|X(0)=E^{\prime}\right\} .
\]
\end{enumerate}
\end{thm}
\begin{proof}
 Recall that, as claimed in Section II, requests arrival according
to a Poisson process with rate $\lambda^{a}$, and the service times
are exponentially distributed with mean $1/\lambda^{c}$, independently
of each other. According to the mining model in Section III, the block
generation is a Poisson process with rate $\lambda_{b}$. In a nutshell,
request inter-arrival times $U^{a}$, block times $U^{b}$, and service
times $U^{c}$ are exponential with mean $1/\lambda^{a}$, $1/\lambda^{b}$,
and $1/\lambda^{c}$, respectively, and independent of each other.

The exponential service times $U^{c}$ imply that if there are $j$
UEs in service, the rate at which service completions occur is $j\lambda^{c}$.
To show this, suppose that $U_{1}^{c},U_{2}^{c},...,U_{j}^{c}$ are
the duration of $j$ i.i.d. exponential simultaneously running time
intervals with mean $1/\lambda^{c}$. Let $U_{min}^{c}=\min\left\{ U_{1}^{c},U_{2}^{c},...,U_{j}^{c}\right\} $
be the minimum service time. Observe that $U_{min}^{c}$ will exceed
$u$ if and only if all $U_{j}^{c}$ exceed $u$. Hence,
\begin{align*}
\text{Pr}\left\{ U_{min}^{c}>u\right\}  & =\text{Pr}\left\{ U_{1}^{c}>u\right\} \cdot\text{Pr}\left\{ U_{2}^{c}>u\right\} \cdots\text{Pr}\left\{ U_{j}^{c}>u\right\} \\
 & =\exp\left(-j\lambda^{c}u\right).
\end{align*}
Only those UEs that are in service can possibly leave. Hence, the
service completion rate with $j$ simultaneous services is $j\lambda^{c}$
for $0\le j\le s$. Since at most $s$ UEs can be in service simultaneously,
obviously the service completion rate is at most $s\lambda^{c}$.
Hence, we denote
\begin{equation}
\lambda_{j}^{c}=\min(j,s)\cdot\lambda^{c}\label{eq:lambdacj}
\end{equation}
 to represent the service completion rate of B-RAN compactly.

\begin{figure*}
\begin{raggedright} \centering\includegraphics[width=0.7\textwidth]{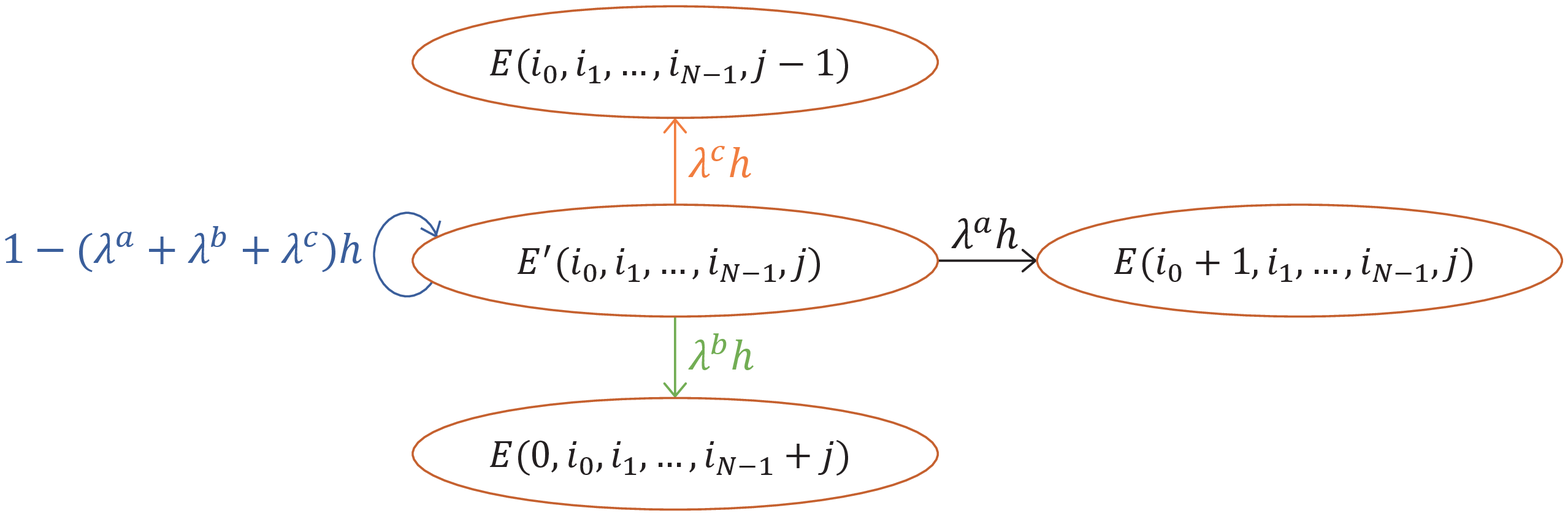}

\end{raggedright} \centering{}\caption{State transition graph of $E\left(i_{0},i_{1},...,i_{N-1},j\right)$.\label{fig:statetransition}}

\vspace{-0.6cm}
\end{figure*}
Similarly, the time until the next event (an request arrival, a block
generation, or a service completion) is also exponentially distributed
with rate $(\lambda^{a}+\lambda^{b}+\lambda_{j}^{c})$. The probability
that an event occurs in $(t,t+h)$ is $1-\exp\left(-\left(\lambda^{a}+\lambda^{b}+\lambda_{j}^{c}\right)h\right)$,
which tends to $\left(\lambda^{a}+\lambda^{b}+\lambda_{j}^{c}\right)h+o(h)$
as $h\rightarrow0$. The length of time required for the event to
occur and the type of the event are independent. There are several
possible changes to a state. If a new request arrives, the state of
B-RAN will switch from $E^{\prime}\left(i_{0},i_{1},...,i_{N-1},j\right)$
to $E\left(i_{0}+1,i_{1},...,i_{N-1},j\right)$. The probability that
a new request arrives in $(t,t+h)$ is
\begin{align}
 & \text{Pr}\left\{ X(t+h)=E\left(i_{0}+1,i_{1},...,i_{N-1},j\right)|X(t)=E^{\prime}\left(i_{0},i_{1},...,i_{N-1},j\right)\right\} \label{queue:a}\\
 & =\frac{\lambda^{a}}{\lambda^{a}+\lambda^{b}+\lambda_{j}^{c}}\left(\left(\lambda^{a}+\lambda^{b}+\lambda_{j}^{c}\right)h+o(h)\right)=\lambda^{a}h+o(h),\;(h\rightarrow0).\nonumber
\end{align}
If a new block is generated, all existing blocks will get one more
confirmation, and the state of B-RAN will move from $E^{\prime}\left(i_{0},i_{1},...,i_{N-1},j\right)$
to $E\left(0,i_{0},i_{1},...,i_{N-1}+j\right)$. The probability that
a block is generated in $(t,t+h)$ can be found to satisfy
\begin{equation}
\lim_{h\rightarrow0}\text{Pr}\left\{ X(t+h)=E\left(0,i_{0},i_{1},...,i_{N-1}+j\right)|X(t)=E^{\prime}\left(i_{0},i_{1},...,i_{N-1},j\right)\right\} =\lambda^{b}h+o(h).\label{queue:b}
\end{equation}
If an access service is ended ($j\geq1$ at this instant), the state
of B-RAN will switch from $E^{\prime}\left(i_{0},i_{1},...,i_{N-1},j\right)$
to $E\left(i_{0},i_{1},...,i_{N-1},j-1\right)$. The probability that
a service completion occurs in $(t,t+h)$ is
\begin{align}
\lim_{h\rightarrow0}\text{Pr}\left\{ X(t+h)=E\left(i_{0},i_{1},...,i_{N-1},j-1\right)|X(t)=E^{\prime}\left(i_{0},i_{1},...,i_{N-1},j\right)\right\}  & =\lambda_{j}^{c}h+o(h).\label{queue:c}
\end{align}
The probability that no event occurs in $(t,t+h)$ is given by
\begin{equation}
\lim_{h\rightarrow0}\text{Pr}\left\{ X(t+h)=X(t)\right\} =1-(\lambda^{a}+\lambda^{b}+\lambda_{j}^{c})h+o(h).\label{queue:noevent}
\end{equation}
The probability that more than one event occurs in $(t,t+h)$ is $o(h)$
as $h\rightarrow0$.

Now we have obtained the transition probabilities $\text{Pr\ensuremath{\left\{  X(t+h)=E|X(t)=E^{\prime}\right\} } }$
for any possible $E$ and $E^{\prime}$ in the state space. Observe
that the transition probabilities are irrelevant to the starting time
$t$, which indicates that $\left\{ X(t),\;t\geq0\right\} $ is homogeneous
in time. Moreover, the transition probabilities are independent of
the states of previous moments, which implies the Markov property.
Therefore, we have proven that $\left\{ X(t),\;t\geq0\right\} $ is
a time-homogeneous Markov process.
\end{proof}
Although the queuing model established by Theorem \ref{thm:queuemodel}
is a vector stochastic process with possibly high dimensions, we would
like to emphasize that such a queuing model is more tractable than
the original non-Markovian process. According to Theorem \ref{thm:queuemodel},
we can directly obtain the transition probabilities, given by Corollary
\ref{cor:transition}.
\begin{cor}
\label{cor:transition}The queuing model $\left\{ X(t),\thinspace t\ge0\right\} $
can be characterized by the transition probabilities $\text{Pr\ensuremath{\left\{ X(h)=E|X(0)=E^{\prime}\right\} }}$.
The cases with non-zero transition probabilities are covered by \eqref{queue:a}-\eqref{queue:noevent},
and probabilities of the rest equal zeros.
\end{cor}
Fig. \ref{fig:statetransition} can demonstrate the transition relationships
in the result of Corollary \ref{cor:transition}, though the case
of $j=0$ is slightly different. According to Corollary \ref{cor:transition},
the basic elements $\left\{ \lambda^{a},\lambda^{b},\lambda^{c},s\right\} $
and the number of confirmations $N$ are enough to determine the transition
probabilities and thus characterize the behaviors of B-RAN. Hence,
we introduce a 4-tuple $\Phi=\left\{ \lambda^{a},\lambda^{b},\lambda^{c},s\right\} $
as basic configurations to describe B-RAN. In the following sections,
we will look at B-RAN and analyze it from a deeper view in more dimensions.

\section{Latency Analysis of B-RAN}

\subsection{Steady-State Analysis}

\begin{figure}
\centering %
\begin{minipage}[t]{0.47\linewidth}%
\begin{raggedright} \centering \includegraphics[width=0.9\textwidth]{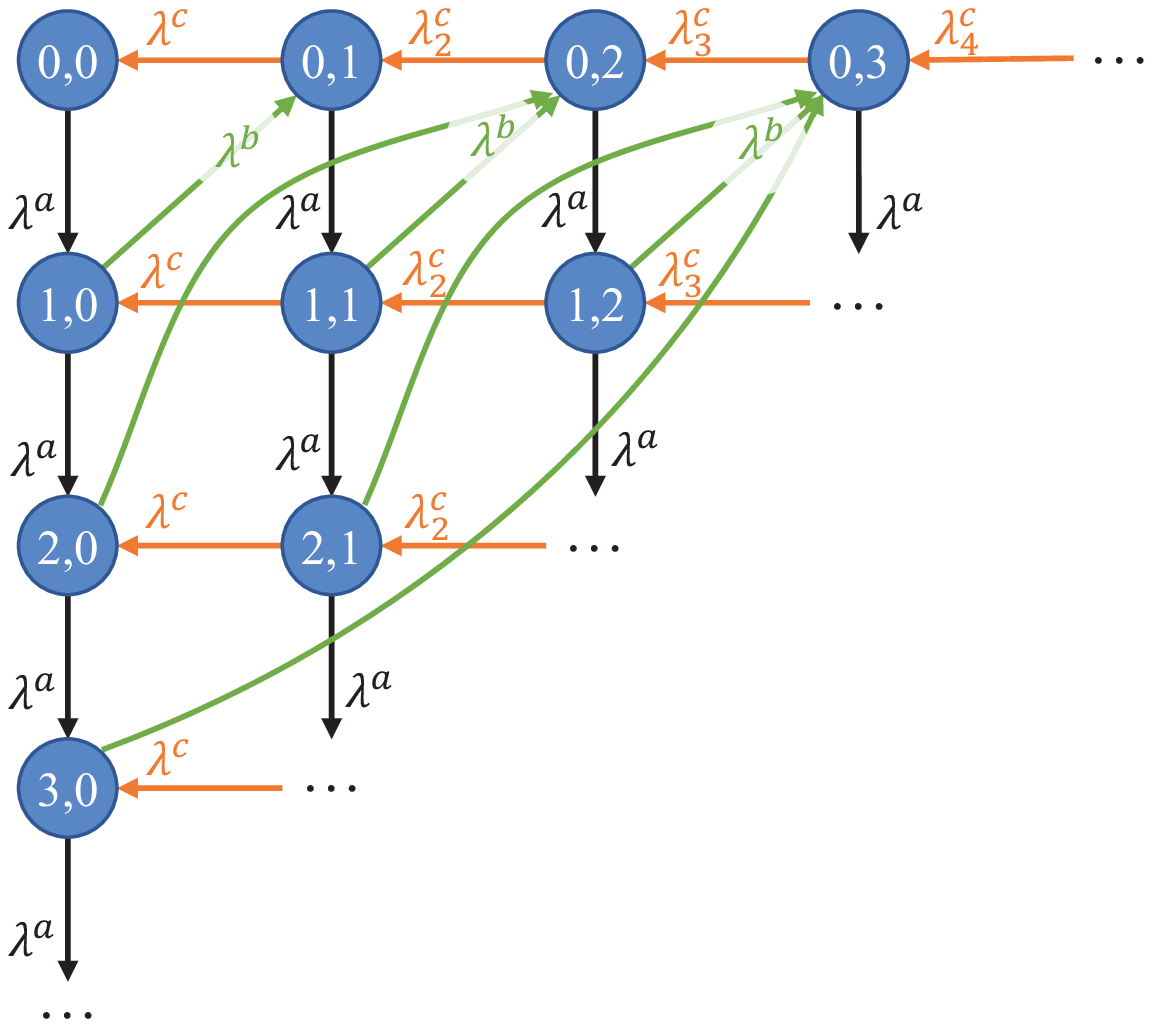}

\end{raggedright} \centering {}\caption{State space diagram in B-RAN.\label{fig:Statespacediagram}}
\end{minipage}\hfill{}%
\begin{minipage}[t]{0.47\linewidth}%
\begin{raggedright} \centering \includegraphics[width=0.9\textwidth]{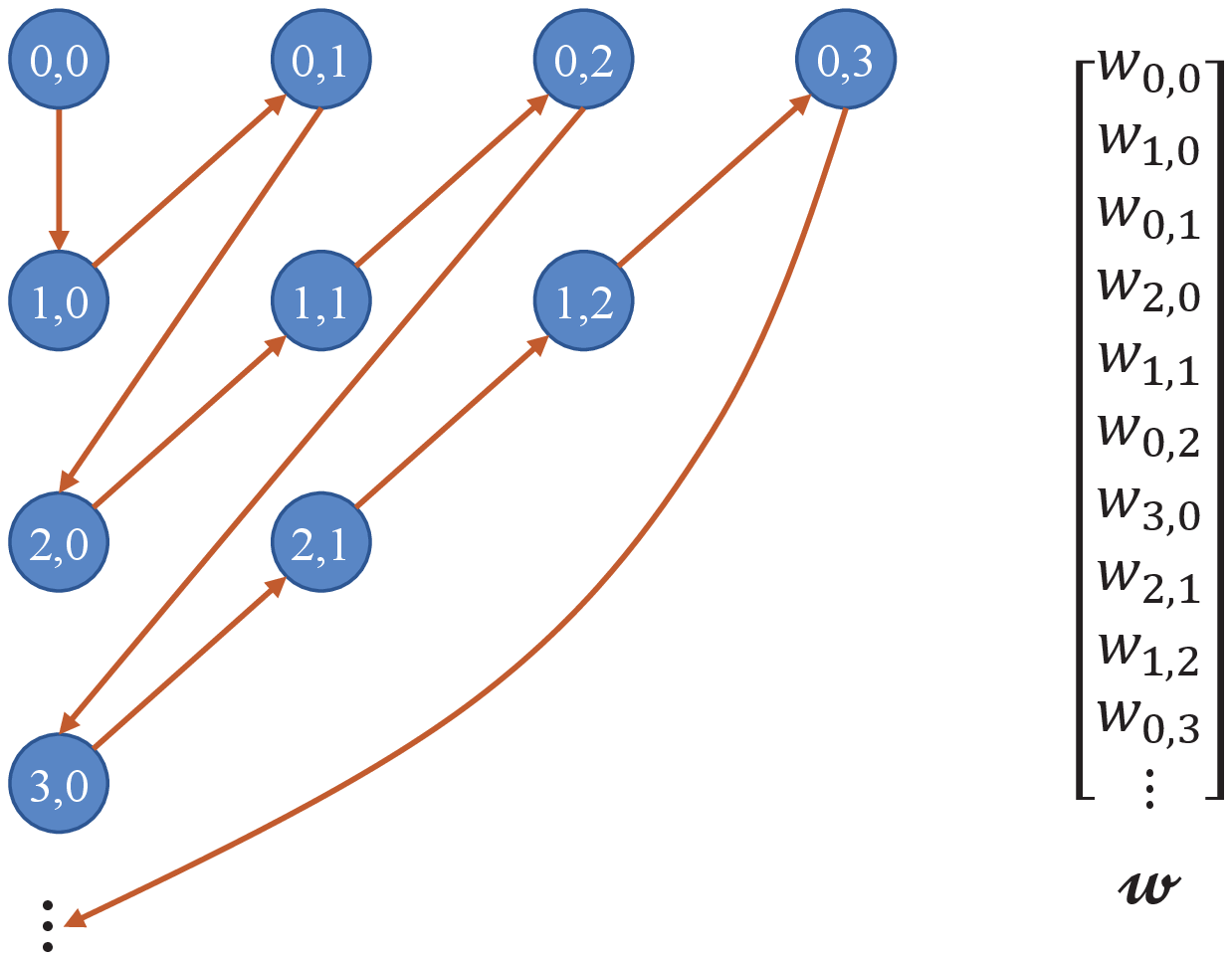}

\end{raggedright} \centering {}\caption{States $E(i,j)$ rearrangement into a row.\label{fig:order}}
\end{minipage}\vspace{-0.6cm}
\end{figure}
Now we have thus far modeled the B-RAN as a time-homogeneous Markov
process with $\left(N+1\right)$ dimensions.  However, the dimensionality
of the state space results in a complex probability transition graph
and is difficult to analyze in general. In the section, we will analyze
access service latency of B-RAN by starting from a relatively simple
one-confirmation case (i.e., only one confirmation is required to
confirm a request). In other words, a request is confirmed, as long
as it is assembled into a block. In the one-confirmation case, the
queuing model is presented by $\left\{ X(t)=E(i_{0},j),\quad t\geq0\right\} $.
We drop the subscript of $i_{0}$ for notational simplicity. State
$E(i,j)$ means that $i$ pending requests are waiting for assembling
into a block and $j$ confirmed requests are waiting for service.
Define $w_{i,j}(t)=\text{Pr}\left\{ X(t)=E(i,j)\right\} $ as the
probability of the queue in state $E(i,j)$ at time $t$. Now let
us investigate the transition probabilities during time $h$ with
the help of Corollary \ref{cor:transition}. By comparing the state
of B-RAN at time $t+h$ with that at time $t$, for all $i=1,2,...$
and $j=0,1,2,...$, we have
\begin{align*}
w_{i,j}(t+h)-w_{i,j}(t) & =w_{i-1,j}(t)\lambda^{a}h+w_{i,j+1}(t)\lambda_{j+1}^{c}h-w_{i,j}(t)\left(\lambda^{a}+\lambda^{b}+\lambda_{j}^{c}\right)h,
\end{align*}
where $\lambda_{j}^{c}$ is defined by \eqref{eq:lambdacj}. By letting
$h\rightarrow0$, we get the following differential-difference equation:
\begin{align*}
\frac{d}{dt}w_{i,j}(t) & =w_{i-1,j}(t)\lambda_{a}+w_{i,j+1}(t)\lambda_{j+1}^{c}-w_{i,j}(t)\left(\lambda_{a}+\lambda_{b}+\lambda_{j}^{c}\right).
\end{align*}
Let $\left\{ w_{i,j}\right\} $ be the steady-state distribution of
B-RAN. The equilibrium condition $\frac{d}{dt}w_{i,j}(t)=0$ yields
\begin{align}
w_{i+1,j}\lambda_{a}+w_{i,j+1}\lambda_{j+1}^{c}-w_{i,j}\left(\lambda_{a}+\lambda_{b}+\lambda_{j}^{c}\right) & =0.\label{latency:diffEqua}
\end{align}
 For the boundary cases ($i=0$), we have
\begin{align}
w_{1,0}\lambda_{1}^{c}-w_{0,0}\lambda^{a} & =0,\label{latency:diffEqub}\\
\left(\sum_{\ell=1}^{j}w_{\ell,j-\ell}\right)\lambda_{b}+w_{0,j+1}\lambda_{j+1}^{c}-w_{0,j}\left(\lambda_{a}+\lambda_{j}^{c}\right) & =0,\;\forall j=0,1,2,...\label{latency:diffEquc}
\end{align}
 The differential-difference equations \eqref{latency:diffEqua}-\eqref{latency:diffEquc}
are known as the forward Kolmogorov equations\cite{Kleinrock1975}.
We illustrate the state transition relationships with one confirmation
in Fig. \ref{fig:Statespacediagram}.

In order to present the forward Kolmogorov equations in a compact
form, we rearrange the two-dimension states $\left\{ w_{i,j}\right\} $
by a particular order as shown in Fig. \ref{fig:order}, captured
by the probability vector:
\[
{\bf w}=\left[\begin{array}{c|cc|ccc|c}
w_{0,0} & w_{1,0} & w_{0,1} & w_{2,0} & w_{1,1} & w_{0,2} & \cdots\end{array}\right]^{T}.
\]
 Now the forward Kolmogorov equations can be rewritten in a matrix
form:
\begin{equation}
{\bf Q}{\bf w}={\bf 0},\label{latency:steadyequation}
\end{equation}
where ${\bf Q}$ is known as the infinitesimal generator, or transition
rate matrix:

\begin{align}
{\bf Q} & =\left[\begin{array}{c|cc|ccc|c}
-\lambda^{a} &  & \lambda_{1}^{c} &  &  &  & \cdots\\
\hline \lambda^{a} & -\left(\lambda^{a}+\lambda^{b}\right) &  &  & \lambda_{1}^{c} &  & \cdots\\
 & \lambda^{b} & -\left(\lambda^{a}+\lambda_{1}^{c}\right) &  &  & \lambda_{2}^{c} & \cdots\\
\hline  & \lambda^{a} &  & -\left(\lambda^{a}+\lambda^{b}\right) &  &  & \cdots\\
 &  & \lambda^{a} &  & -\left(\lambda^{a}+\lambda^{b}+\lambda_{1}^{c}\right) &  & \cdots\\
 &  &  & \lambda^{b} & \lambda^{b} & -\left(\lambda^{a}+\lambda_{2}^{c}\right) & \cdots\\
\hline \vdots & \vdots & \vdots & \vdots & \vdots & \vdots & \ddots
\end{array}\right].
\end{align}
The entry in ${\bf Q}$ equals to the corresponding transition rate
given by $\frac{d}{dh}\text{Pr}\left\{ X(h)=E|X(0)=E^{\prime}\right\} $
only depending on the B-RAN configuration tuple $\Phi=\left\{ \lambda^{a},\lambda^{b},\lambda^{c},s\right\} $.
We can numerically solve the matrix equation by combining with the
sum probability condition of ${\bf 1}^{T}{\bf w}=1$, i.e.,
\begin{equation}
\left[\begin{array}{c}
{\bf Q}\\
{\bf 1}^{T}
\end{array}\right]{\bf w}=\left[\begin{array}{c}
{\bf 0}\\
1
\end{array}\right].\label{latency:matrix}
\end{equation}
From \eqref{latency:matrix}, the steady-state distribution ${\bf w}(\Phi)$
can be expressed as an implicit function of $\Phi$. Note that the
waiting space of B-RAN has no maximum limit. The number of states,
i.e., the dimensions of the vector ${\bf w}$, should be infinite.
In numerical calculations, we can use the solution with large enough
but finite dimensions to approximate the infinite-dimension one. However,
in practice the number of UEs in a tract cannot be infinite, either.
The aggressive load $\lambda^{a}$ is required to be less than $\lambda^{c}$
by the stable condition.

We can obtain the steady-state distribution of B-RAN via \eqref{latency:matrix}.
Meanwhile, we use our self-built prototype to measure the sojourn
time of each state and estimate the probability of a state. The results
show that analytical steady distributions are highly consistent with
experimental outcomes, thereby validating our established queuing
model. We illustrate the steady-state distributions of B-RAN with
different $T^{b}$ and different traffic intensities $\rho=\frac{\lambda^{a}}{s\lambda^{c}}$
in Fig. \ref{fig:distributionSamples}. The low, medium and high traffic
intensities $\rho=\frac{\lambda^{a}}{s\lambda^{c}}$ are set to 0.1,
0.4 and 0.7, respectively. Each node represents $E(i,j)$ in the figures.
The results show that, on the one hand, under a higher traffic intensity,
there are more confirmed requests waiting for service. On the other
hand, when the block time $T^{b}$ becomes larger, more requests wait
to be assembled into a block.

\begin{figure}
\begin{raggedright} \centering\subfigure[]{ \includegraphics[width=0.3\textwidth]{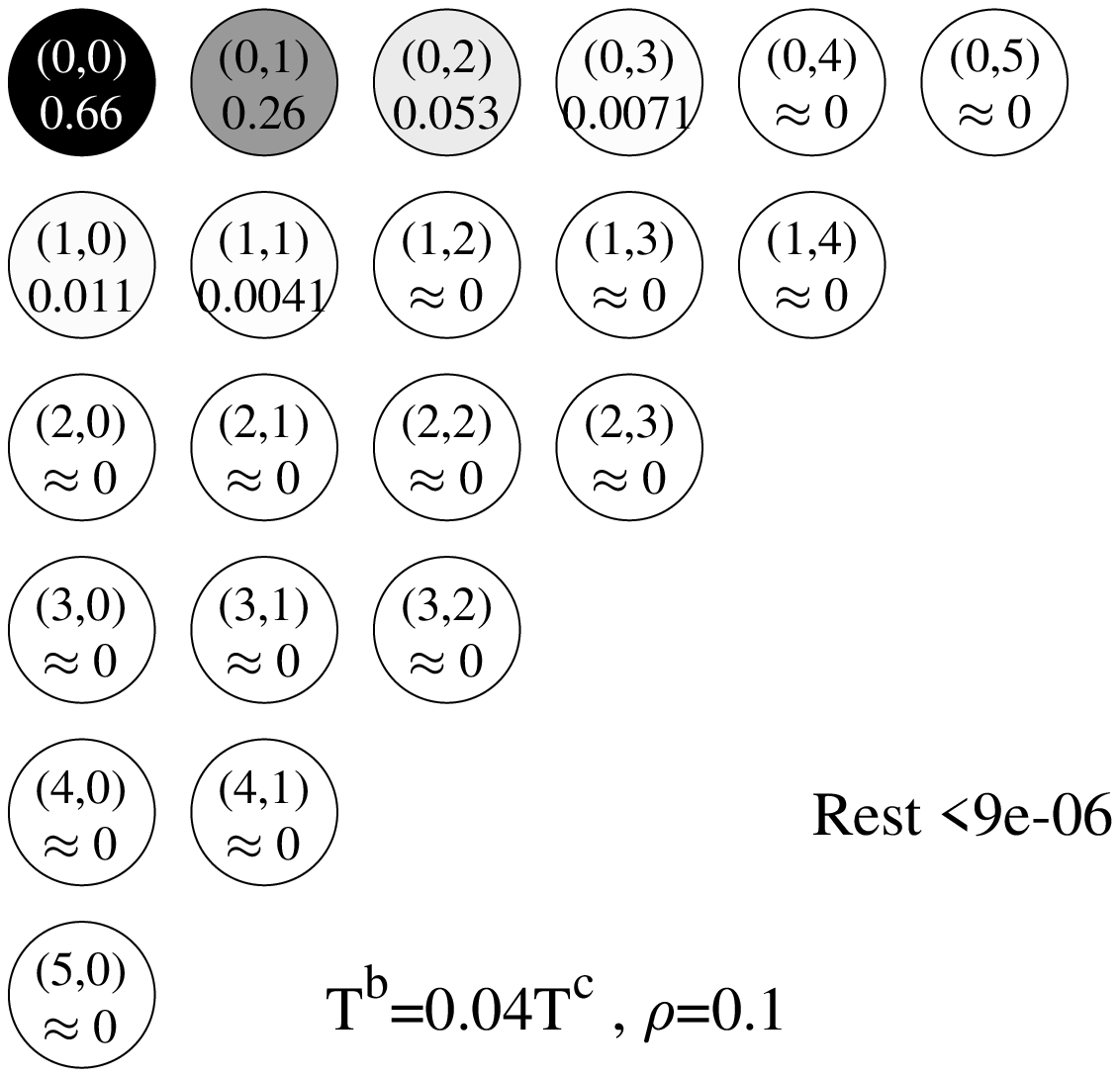}}
\hfill{}\subfigure[]{ \includegraphics[width=0.3\textwidth]{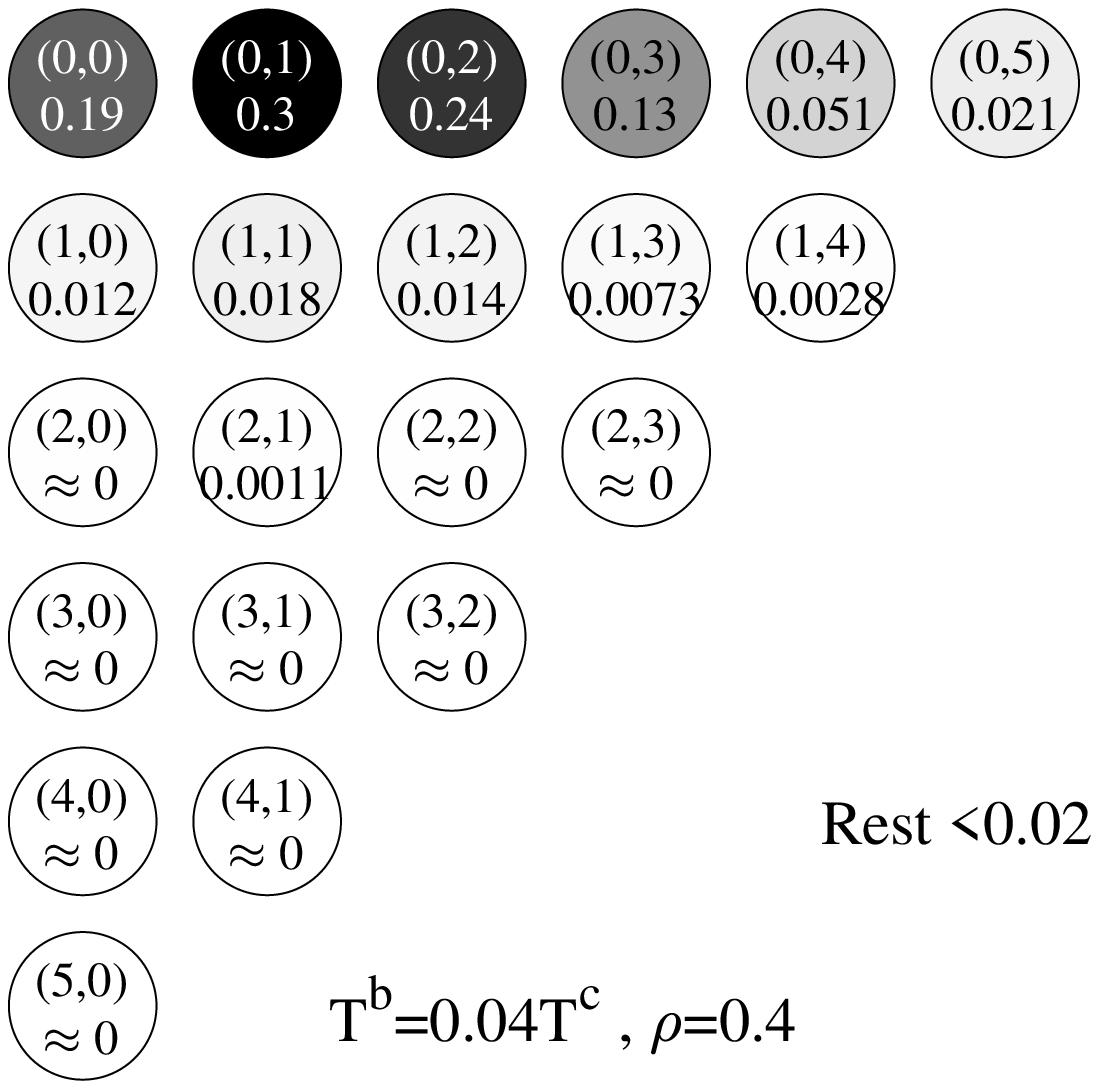}}\hfill{}\subfigure[]{
\includegraphics[width=0.3\textwidth]{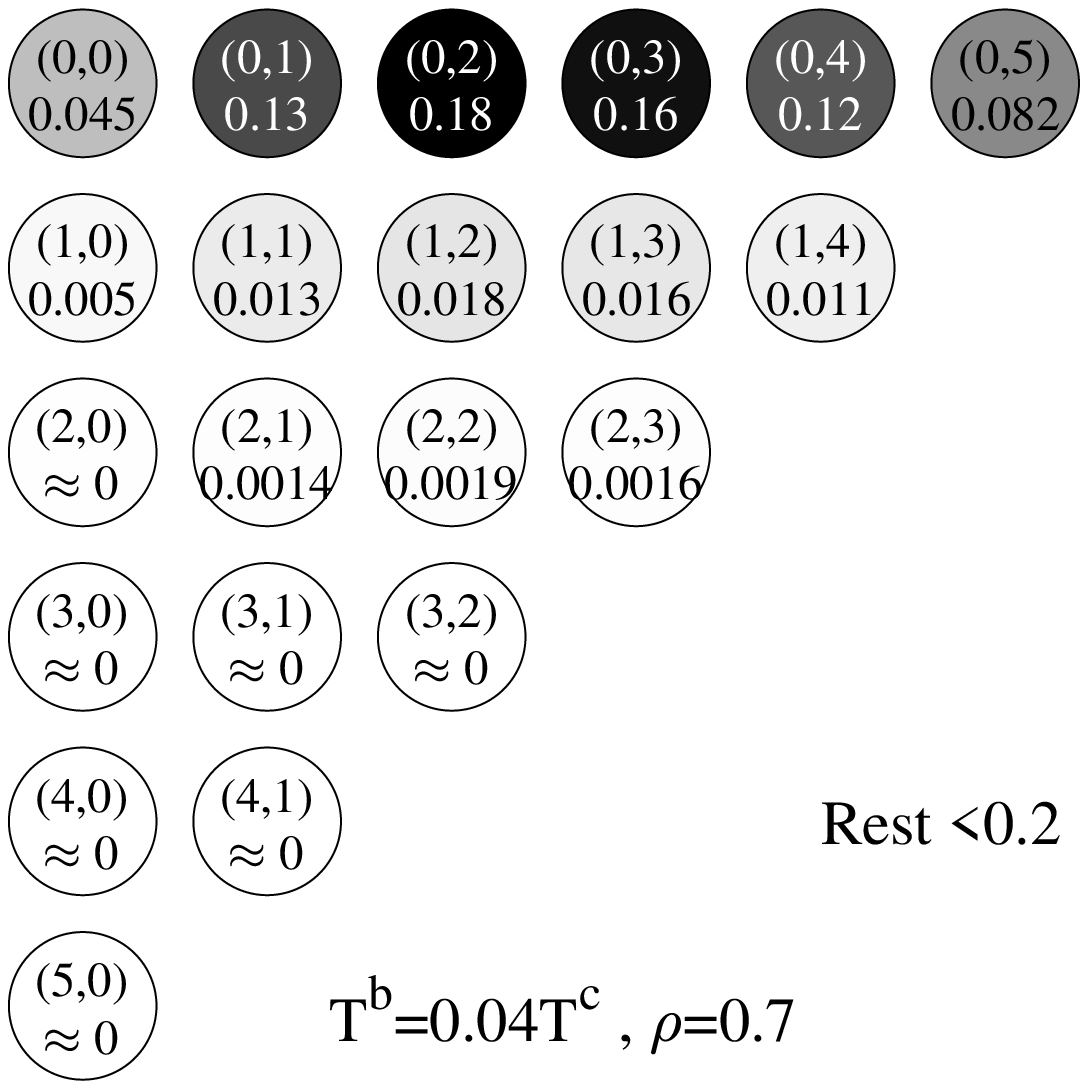}}
\hfill{} \subfigure[]{ \includegraphics[width=0.3\textwidth]{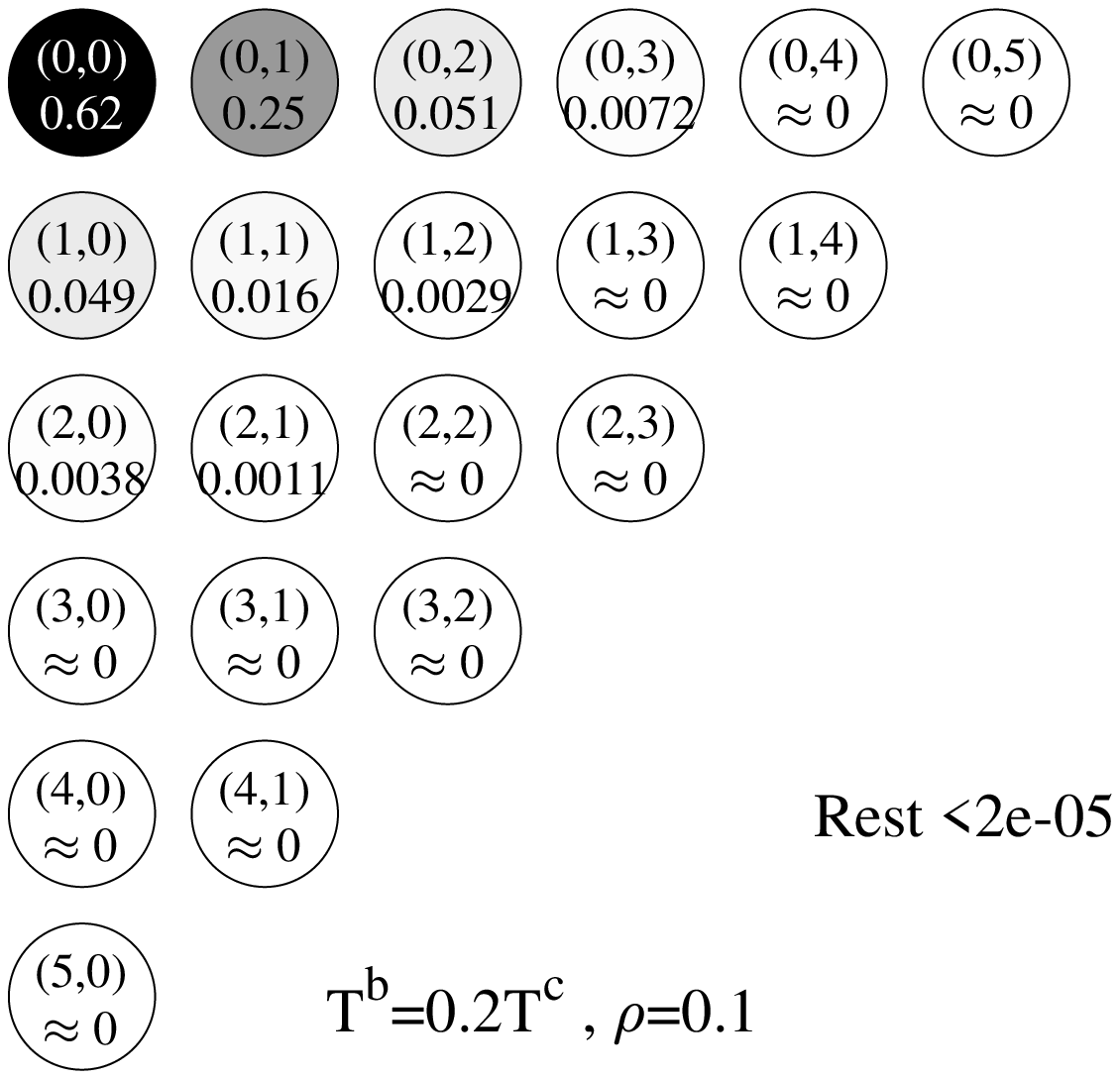}}
\hfill{}\subfigure[]{ \includegraphics[width=0.3\textwidth]{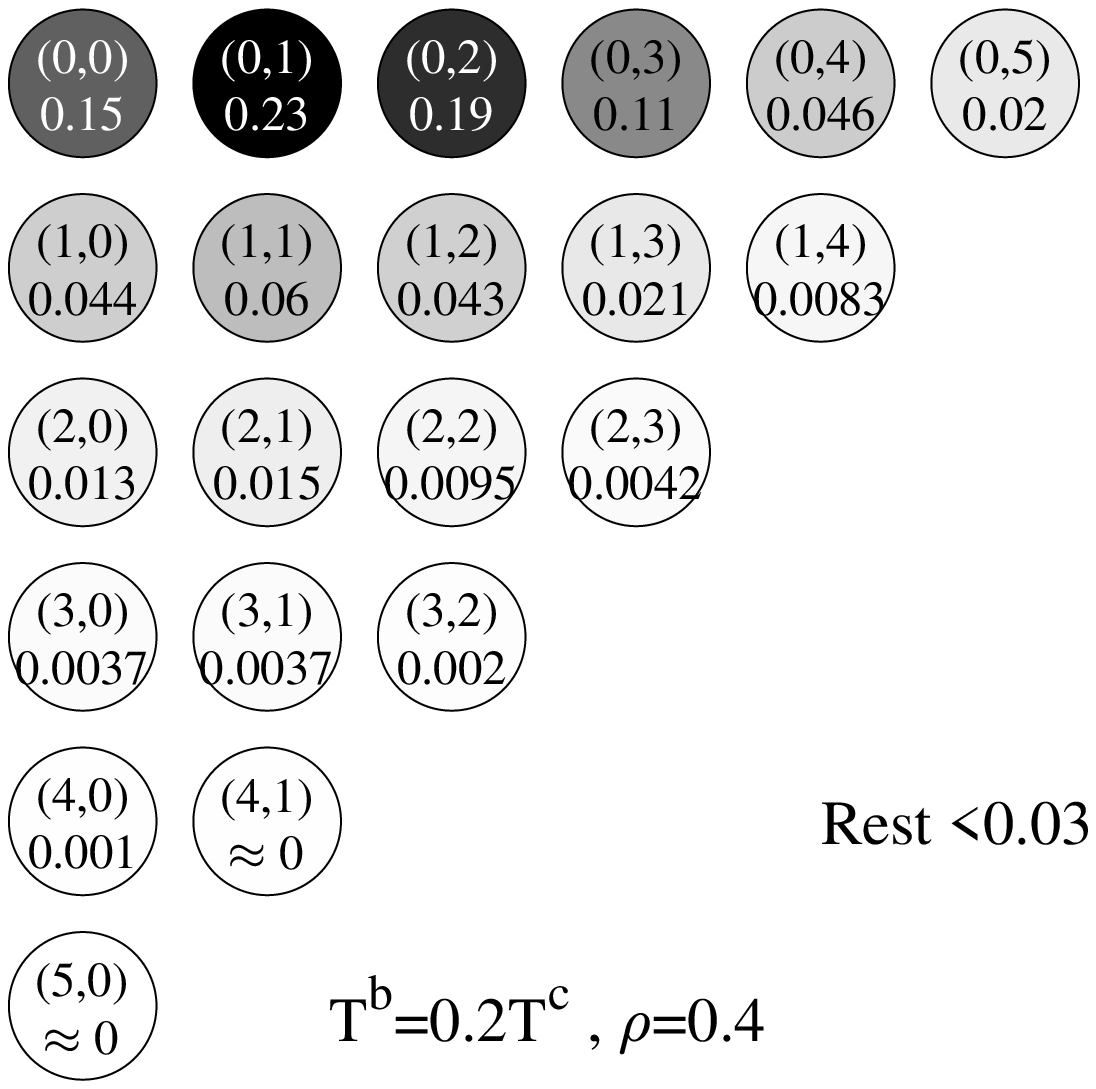}}\hfill{}\subfigure[]{
\includegraphics[width=0.3\textwidth]{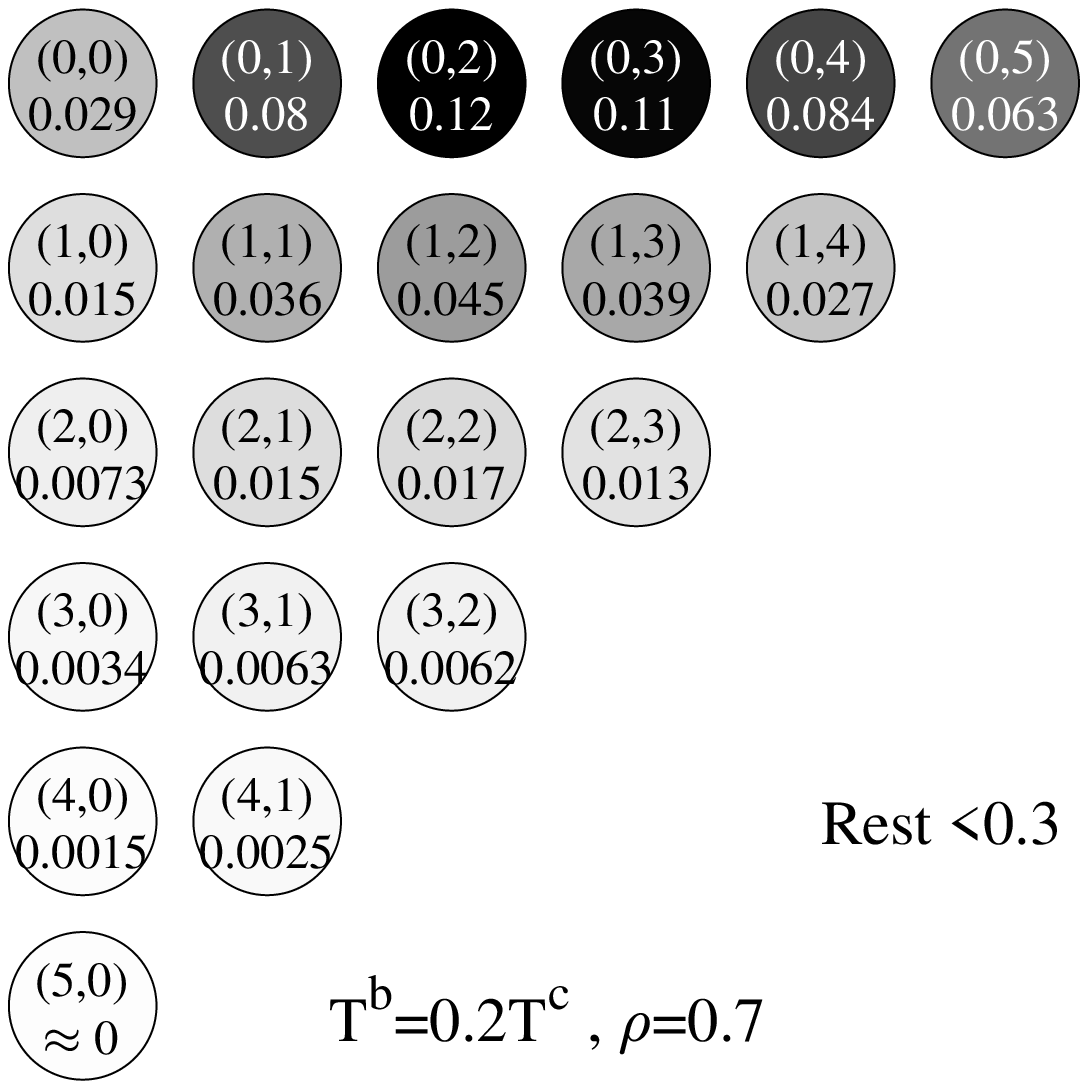}}

\end{raggedright} \centering {}\caption{The distribution of steady states under different traffic intensities
and block time with $s=4$ links. (a) $T^{b}=0.04T^{c}$ under low
traffic intensity $\rho=0.1$. (b) $T^{b}=0.04T^{c}$ under medium
traffic intensity $\rho=0.4$. (c) $T^{b}=0.04T^{c}$ under high traffic
intensity $\rho=0.7$. (d) $T^{b}=0.2T^{c}$ under low traffic intensity
$\rho=0.1$. (e) $T^{b}=0.2T^{c}$ under medium traffic intensity
$\rho=0.4$. (f) $T^{b}=0.2T^{c}$ under high traffic intensity $\rho=0.7$.
\label{fig:distributionSamples}}
\vspace{-0.6cm}
\end{figure}

To analyze the average latency in B-RAN, we consider the limiting
distribution ${\bf w}(\Phi)$, from which we can obtain the average
number of waiting requests $\mathsf{N}\left(\Phi\right)$ in B-RAN:
\begin{align}
\mathbb{E}\left\{ \mathsf{N}\left(\Phi\right)\right\}  & =\sum_{i,j}\left(i+j\right)\cdot w_{i,j}\left(\Phi\right).
\end{align}
We can apply the Little's Law\cite{Cooper1972} as a bridge to connect
the expected queue length and the average latency. The Little\textquoteright s
Law states that, in a stable system, the average number of items is
equal to the arrival rate multiplied by the average time an item spends
in the system. Hence, the expected sojourn time $\mathsf{L}_{s}\left(N,\Phi\right)$
is the sum of waiting time and service time (all the four stages of
the workflow), and the expression of expected sojourn time in the
one-confirmation case is given by:

\begin{equation}
\mathsf{L}_{s}\left(N=1,\Phi\right)=\mathbb{E}\left\{ \mathsf{N}\left(\Phi\right)\right\} /\lambda^{a}=T^{a}\sum_{i,j}\left(i+j\right)w_{i,j}\left(\Phi\right).
\end{equation}
According to Section II, the expected service time is $T^{c}$. Thus,
the average latency $\mathsf{L}\left(N,\Phi\right)$ is defined as
the expected waiting time (the first three stages). In the one-confirmation
case, $\mathsf{L}\left(N,\Phi\right)$ can be written as\footnote{Another equivalent expression of expected latency is $\mathsf{L}\left(N=1,\Phi\right)=T^{a}\sum_{i,j}\left(i+\left(j-1\right)^{+}\right)w_{i,j}\left(\Phi\right)$,
where $\left(\cdot\right)^{+}=\max\left\{ \cdot,0\right\} $.}
\begin{equation}
\mathsf{L}\left(N=1,\Phi\right)=T^{a}\sum_{i,j}\left(i+j\right)w_{i,j}\left(\Phi\right)-T^{c}.\label{latency:1cmf}
\end{equation}
\eqref{latency:1cmf} gives the average latency with one confirmation
in terms of the limiting distribution in \eqref{latency:steadyequation}.

However, \eqref{latency:1cmf} only applies the one-confirmation special
case. Our goal is to investigate the general $N$-confirmation problem.
It is difficult to directly analyze a queue with $\left(N+1\right)$-dimensional
state space, mainly due to the excessively large number variables
in solving the equilibrium equations \eqref{latency:steadyequation}.
Recall that we are more interested in the average system latency.
 A request, once assembled into a block, is required to wait for
$N-1$ confirmations after the very first confirmation of assembling
into a block. This waiting period is an additional stage compared
to the one-confirmation case. The additional $N-1$ confirmations
correspond to extra waiting time of $N-1$ independent block time
$U_{n}^{b}$. Hence, the average access latency is given by
\begin{align}
\mathsf{L}\left(N,\Phi\right) & =\mathsf{L}\left(1,\Phi\right)+\mathbb{E}\left\{ \sum_{n=2}^{N}U_{n}^{b}\right\} \nonumber \\
 & =T^{a}\sum_{i,j}\left(i+j\right)w_{i,j}\left(\Phi\right)+T^{b}\left(N-1\right)-T^{c}.\label{latency:expression}
\end{align}
  Despite the implicit expression of $w_{i,j}\left(\Phi\right)$,
one can see that  latency $\mathsf{L}\left(N,\Phi\right)$ grows linearly
with the number of confirmations $N$. Each extra confirmation leads
to $T^{b}$ longer waiting time on average, and fewer confirmations
will effectively reduce service latency. Note that we cannot simply
conclude that  latency is linear in $T^{b}$, because $T^{b}$ will
influence the limiting distribution $w_{i,j}\left(\Phi\right)$, which
in turn also affects latency. For large $N$, however, we can conclude
that  latency is quasi-linear in $T^{b}$.\textbf{}

\begin{figure}
\centering %
\begin{minipage}[t]{0.47\linewidth}%
\begin{raggedright} \centering\includegraphics[width=0.99\textwidth]{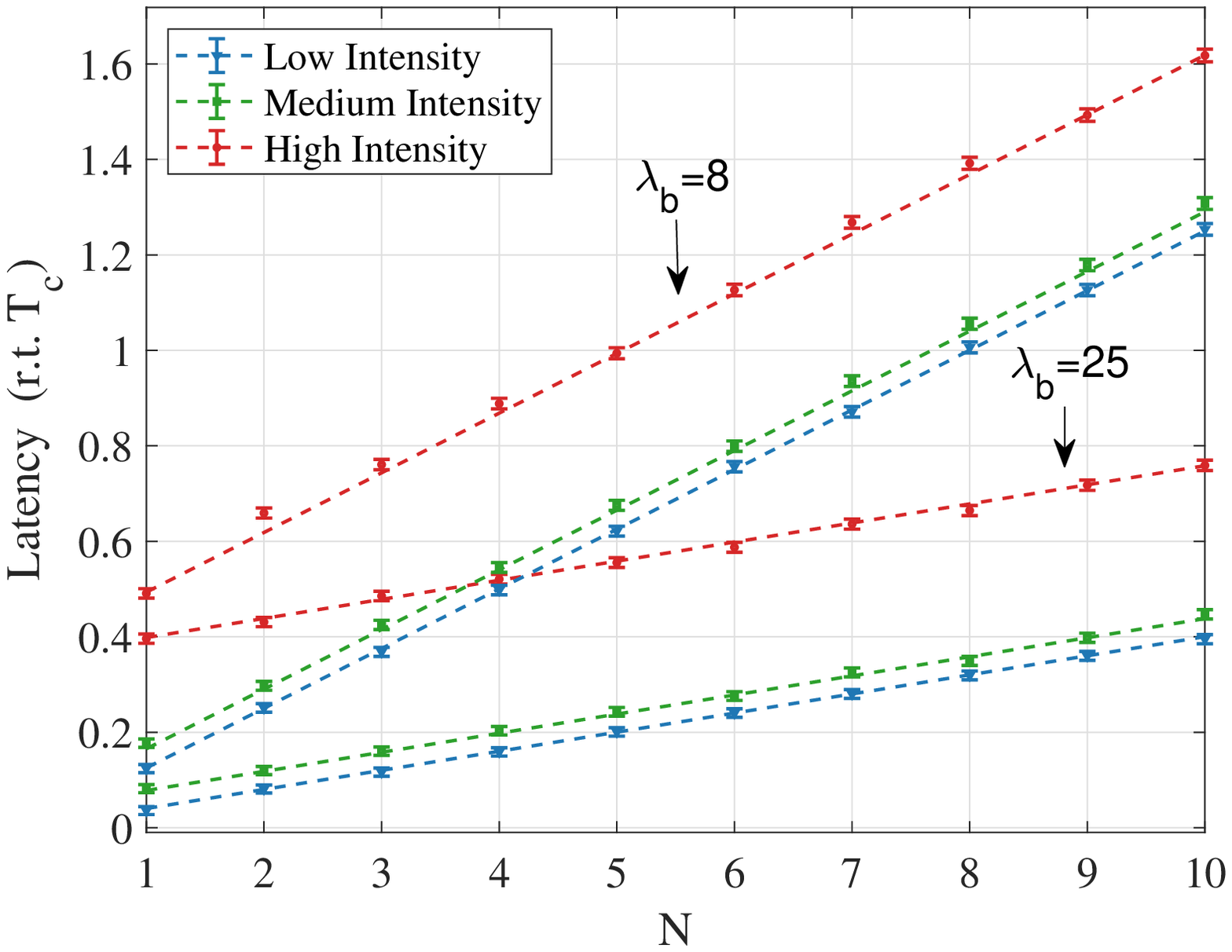}\end{raggedright}
\centering{}\caption{The analytical and experimental latency for different $N$ and $\lambda^{b}$
with $s=4$.\label{fig:latency-tbn}}
\end{minipage}\hfill{}%
\begin{minipage}[t]{0.47\linewidth}%
\begin{raggedright} \centering\includegraphics[width=0.99\textwidth]{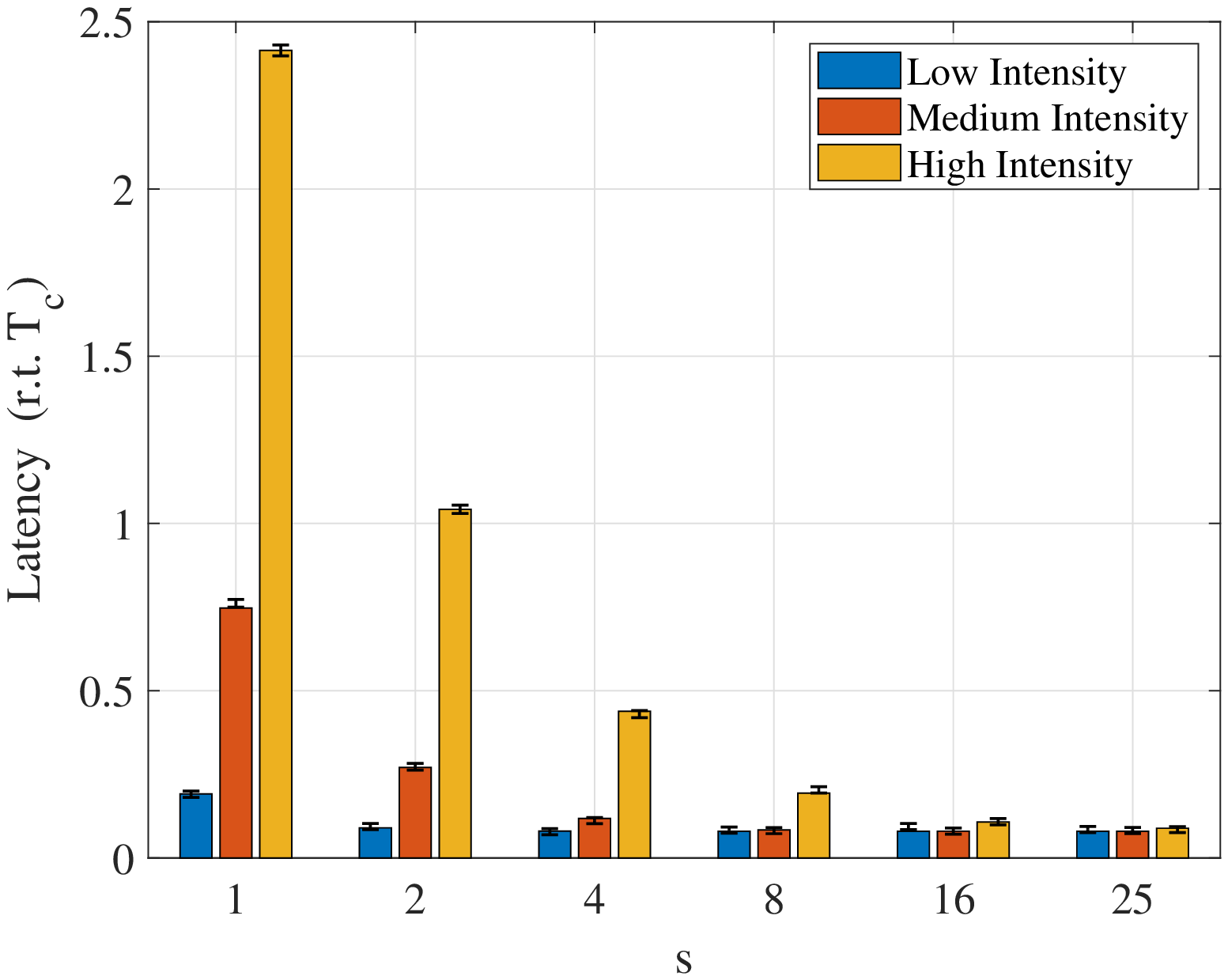}\end{raggedright}
\centering{}\caption{The impact of the maximal access channels $s$ on service latency
($N=2$ and $\lambda^{b}=25$).\label{fig:latency-s}}
\end{minipage}\vspace{-0.6cm}
\end{figure}
In \textbf{Prototype Verification B}, we compare the analytical results
of \eqref{latency:expression} with the outcomes from our B-RAN prototype
under different settings of $N$ and $\lambda^{b}$. Fig. \ref{fig:latency-tbn}
shows both the average and the 95\% confidence interval of latency
under different traffic intensities. In the figure, service latency
is presented relative to (r.t.) the service time $T_{c}$.  First,
the experimental results are highly consistent with the analysis of
\eqref{latency:expression}. Basically, higher traffic intensity will
lead to longer service latency. Furthermore, the results show a strong
linear relationship between latency $\mathsf{L}\left(N,\Phi\right)$
and confirmation number $N$, for which the linear slope is close
to $T^{b}$. This conclusion holds for different traffic intensities
and is a simple corollary of \eqref{latency:expression}: each extra
confirmation leads to exactly $T^{b}$ longer waiting time on average.

Moreover, we can assess the impact of $s$ on the average latency
by Fig. \ref{fig:latency-s}. Given $N=2$ and $\lambda^{b}=25$,
as the number of access links $s$ increases, the average service
latency $\mathsf{L}\left(N,\Phi\right)$ becomes shorter, especially
for the high-intensity case. This result implies that a more extensively
formed B-RAN is less likely to be congested under a constant traffic
$\rho=\frac{\lambda^{a}}{s\lambda^{c}}$. B-RAN through inter-operative
network integration and coordination can benefit from the economy
of scale by delivering shorter and jointly becomes more valuable.
This remarkable insight strongly motivates the future development
of B-RAN.

\subsection{Bounds on Latency}

Although we have an accurate queuing model $\left\{ X(t),\thinspace t\ge0\right\} $
for B-RAN, the relationship between latency in \eqref{latency:expression}
and B-RAN parameters $\Phi$ remains unspecified. Here we reveal the
impact of key parameters on the B-RAN latency by deriving tight bounds.

We first investigate the upper bound of latency by introducing an
extra constraint that each block can carry at most one request. The
ultra-small block size constraint largely affects network throughput
but can effectively simplify the queuing process. For the one-confirmation
case, requests are separately organized into different blocks. Requests
and blocks arrive as a Poisson Process with rate $\lambda^{a}$ and
$\lambda^{b}$, respectively. Hence, the assembling process becomes
an M/M/1 queue with arrival rate $\lambda^{a}$ and service rate $\lambda^{b}$.
After being assembled into blocks, requests get confirmed and wait
for service, effectively forming an M/M/$s$ queue.  Hence, by introduced
the block size limit, the whole process is divided into two queues
in tandem and can be analyzed separately.

Given arrival rate $\lambda$ and service rate $\mu$, the average
waiting time of an M/M/1 queue is $\frac{1}{\mu-\lambda}$\cite{Cooper1972}.
However, the latency of an M/M/$s$ queue is relatively complex, captured
by a known lemma in Chapter 3.4 of \cite{Cooper1972} as follows:
\begin{lem}
Given arrival rate $\lambda$ and service rate $\mu$, the average
waiting time of an M/M/$s$ system is given by $\frac{C(s,\lambda/\mu)}{s\mu-\lambda}+\frac{1}{\mu}$,
where $C(s,\lambda/\mu)$ is the Erlang C formula expressed as
\[
C(s,\lambda/\mu)={\displaystyle \frac{\frac{\left(\lambda/\mu\right)^{s}}{s!}\frac{s}{s-\lambda/\mu}}{\frac{\left(\lambda/\mu\right)^{s}}{s!}\frac{s}{s-\lambda/\mu}+\sum\limits _{i=0}^{s-1}\frac{\left(\lambda/\mu\right)^{i}}{i!}}.}
\]
\end{lem}
In the $N$-confirmation case, we should take the additional confirmation
latency $\frac{N-1}{\lambda^{b}}$ into account.  The overall latency
experienced by a request aggregates the four stages shown in Fig.
\ref{fig:workflow}. Consequently, under the block size limit, the
system latency is upper bounded by
\begin{equation}
\mathsf{L}_{ub}\left(N,\Phi\right)=\frac{1}{\lambda^{b}-\lambda^{a}}+\frac{C(s,\lambda^{a}/\lambda^{c})}{s\lambda^{c}-\lambda^{a}}+\frac{N-1}{\lambda^{b}}.\label{latency:Lub}
\end{equation}
In practice, the block size limit is always much greater than one,
and apparently,  the obtained bound $\mathsf{L}_{ub}\left(N,\Phi\right)$
is much longer than the actual latency. Note that $\mathsf{L}_{ub}\left(N,\Phi\right)$
is available if only $\lambda^{b}>\lambda^{a}$. Otherwise, the length
of the first stage (i.e., the M/M/1 queue) will grow without limits
and $\mathsf{L}_{ub}\left(N,\Phi\right)$ would tend to infinity.

To derive the lower bound of  latency, consider that blocks are rapidly
generated such that a request is assembled into a block immediately
upon its arrival. As a result, the queuing process with one confirmation
is directly reduced to only an M/M/$s$ queue. By considering confirmation
latency, we obtain an lower bound of service latency:
\begin{align}
\mathsf{L}_{lb1}\left(N,\Phi\right) & =\frac{C(s,\lambda^{a}/\lambda^{c})}{s\lambda^{c}-\lambda^{a}}+\frac{N-1}{\lambda^{b}}.\label{latency:Llb1}
\end{align}
The gap between $\mathsf{L}_{lb1}\left(N,\Phi\right)$ and $\mathsf{L}_{ub}\left(N,\Phi\right)$
equals $\frac{1}{\lambda^{b}-\lambda^{a}}$, which is negligible for
$\frac{1}{\lambda^{b}-\lambda^{a}}\rightarrow0$. Therefore, in this
case of $\lambda^{b}\gg\lambda^{a}$ or more accurately $\lambda^{b}-\lambda^{a}\rightarrow\infty$,
we can say both $\mathsf{L}_{lb1}\left(N,\Phi\right)$ and $\mathsf{L}_{ub}\left(N,\Phi\right)$
are quite tight in terms of latency.

Moreover, we can obtain another lower bound by estimating how long
it takes to assemble a request into a block. As shown in Section V-A,
the closed-form distribution of $E\left(i,j\right)$ is intractable.
However, it is possible to deal with the distribution of $i$ in state
$E\left(i,j\right)$. We can sketch the state space diagram of  the
number of pending requests $i$, and then it forms a typical M/M/$\infty$
queue with arrival rate $\lambda^{a}$ and service rate $\lambda^{b}$.
By formulating the forward Kolmogorov equations, we can obtain the
limiting distribution as
\[
\text{Pr}\left\{ X(t)=E(i,\cdot)\right\} =\frac{\lambda^{b}}{\lambda^{a}+\lambda^{b}}\cdot\left(\frac{\lambda^{a}}{\lambda^{a}+\lambda^{b}}\right)^{i}.
\]
 According to the Little's Law, the mean waiting time for assembling
is $1/\lambda^{b}$. By simply neglecting the waiting time for service
(Phase 3), we can attain another lower bound of latency:
\begin{equation}
\mathsf{L}_{lb2}\left(N,\Phi\right)=\frac{1}{\lambda^{b}}+\frac{N-1}{\lambda^{b}}=\frac{N}{\lambda^{b}}.\label{latency:Llb2}
\end{equation}
Since the waiting time for service (phase 3) to devise $\mathsf{L}_{lb2}\left(N,\Phi\right)$
is neglected, the second lower bound can be regarded as an extreme
case under an ultra-light traffic intensity. This is the case when
the request arrival rate $\lambda^{a}$ is extremely small such that
there always exist idle access links, in which case the confirmed
requests no longer have to wait for service. Thus, $\mathsf{L}_{lb2}\left(N,\Phi\right)$
becomes tight as $\lambda^{a}\rightarrow0$. Observe that the arrival
rate $\lambda^{a}$ does not appear in the expression of $\mathsf{L}_{lb2}\left(N,\Phi\right)$,
indicating that $\mathsf{L}_{lb2}\left(N,\Phi\right)$ is a unified
bound for different arrival rates and could be a bit loose in some
cases.

In \textbf{Prototype Verification C}, we would like to evaluate the
accuracy of the queuing model and the bounds through experimental
results. We set the maximum number of access links to $s=25$ and
obtain average and 95\% confidence interval of service latency for
different system setups $\Phi$, as shown in Fig. \ref{fig:latencyBounds}.
From the results in Fig. \ref{fig:latencyBounds}, we see that the
queuing model established in our work accurately predicted the experimental
results. When the traffic intensity is less than $\frac{1}{2}$, the
average latency is dominated by the confirmation delay ($NT^{b}$).
When the traffic intensity becomes larger, the network congest occurs
and leads to a much longer delay. We also provided the upper and lower
bounds in Fig. \ref{fig:latencyBounds}. It is clear that both upper
and lower bounds can help estimate the range of latency. At low traffic
intensity, the average service latency is close to the lower bound
$\mathsf{L}_{lb2}\left(N,\Phi\right)$, where the service latency
is dominated by the confirmation delay, as discussed in Section V-B.
From Fig. \ref{fig:latencyBounds}(b), since $\lambda^{b}$ is much
larger than $\lambda^{a}$, we can see that $\mathsf{L}_{ub}\left(N,\Phi\right)$
and $\mathsf{L}_{lb1}\left(N,\Phi\right)$ become much tighter.

\begin{figure}
\begin{raggedright} \centering\subfigure[]{\includegraphics[width=0.45\textwidth]{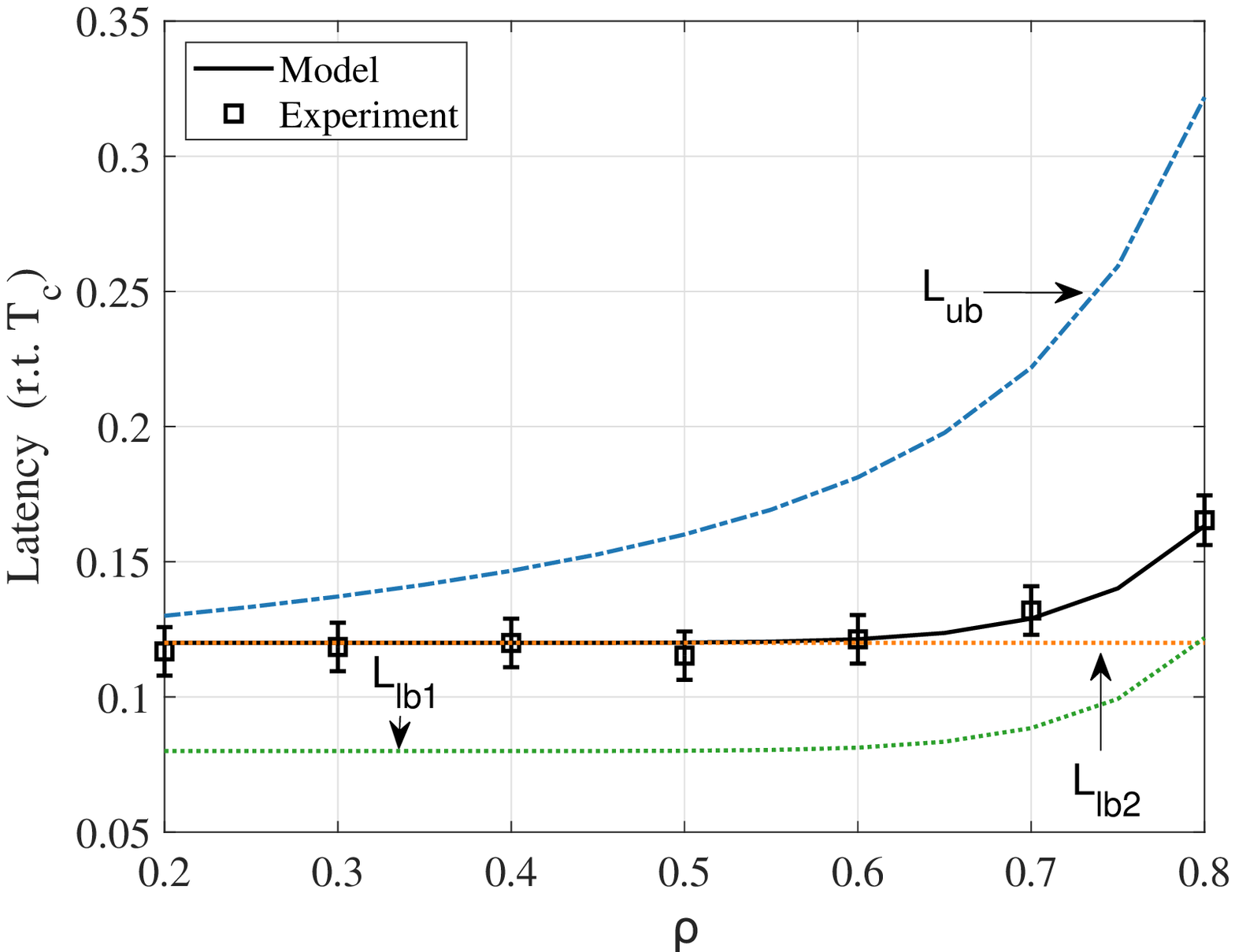}}\hfill{}
\subfigure[]{\includegraphics[width=0.45\textwidth]{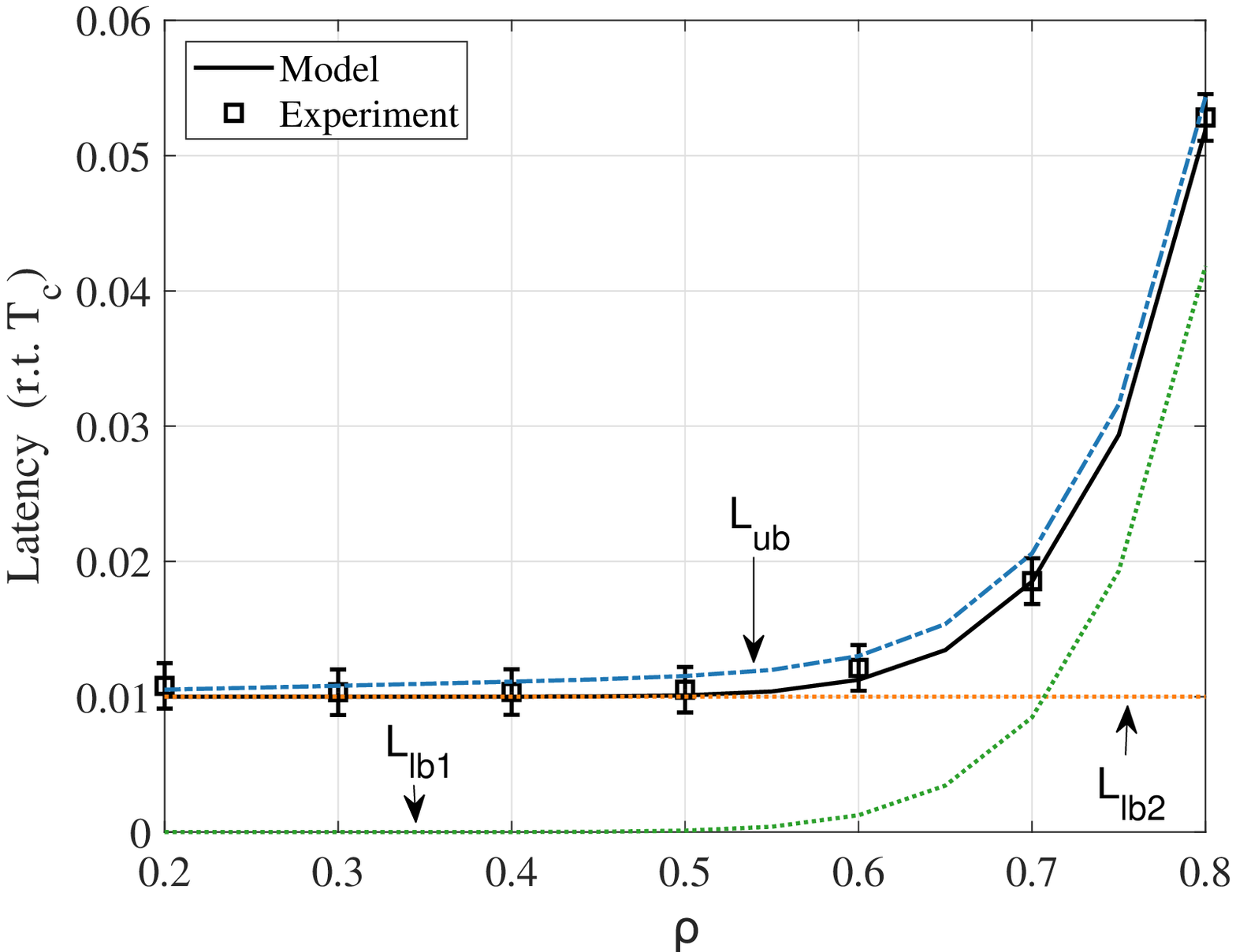}}\end{raggedright}
\centering{}\caption{The analytical and experimental latency with upper and lower bounds
under different traffic intensities $\rho$ ($s=25$). (a) $\lambda^{b}=25$
and $N=3$. (b) $\lambda^{b}=100$ and $N=1$. \label{fig:latencyBounds}}
\vspace{-0.6cm}
\end{figure}

\section{Security Considerations}

\subsection{Alternative History Attack}

Since B-RAN is a distributed radio access network enabled by blockchains,
which do elicit certain security concerns and risks. There can be
a variety of vulnerabilities and security implications in blockchain
systems. The alternative history attack, also referred to as ``double
spending attack'', is a fundamental and inherent risk of distributed
systems\cite{Tschorsch2016,Bai2019}.

The alternative history attack implies that the history record in
blockchain could be tampered by malicious miners, which essentially
exposes the vulnerability of blockchain. In alternative history attack,
an attacker privately mines an alternative blockchain fork in which
a fraudulent double spending event (e.g., an access service in B-RAN)
is included. After the network accepts a benign chain, the attacker
releases the fraudulent fork. If the fraudulent fork eventually becomes
longer than the benign one, the attacker can successfully alter a
confirmed history and spend the same coin (credit) twice, which would
drive the blockchain system into a catastrophic inconsistent state.
 For wireless network, altering a confirmed chain in B-RAN may lead
to interference issues when more than one UEs attempt to access the
same channel simultaneously. Hence, alternative history attack cannot
be simply ignored, despite the difference transactions between B-RAN
and cryptocurrencies. In our security considerations, we focus mainly
on the risk of alternative history attacks.

For blockchain to overcome alternative history attack, miners are
always guided to the \textquotedblleft longest\textquotedblright{}
locally known fork, that is, the one involving the highest amount
of computational effort so far. The block generation process is coupled
to some specific capability of miners, e.g., computational power in
PoW and the number of available devices in PoD. Still, alternative
history attacks are possible, and only require adversaries to create
a chain longer than the benign one. Clearly, the higher the number
of confirmations, the lower the probability of a successful alternative
history attack.

\subsection{Probability of Successful Attack}

\begin{figure}
\begin{raggedright} \centering\includegraphics[width=0.8\textwidth]{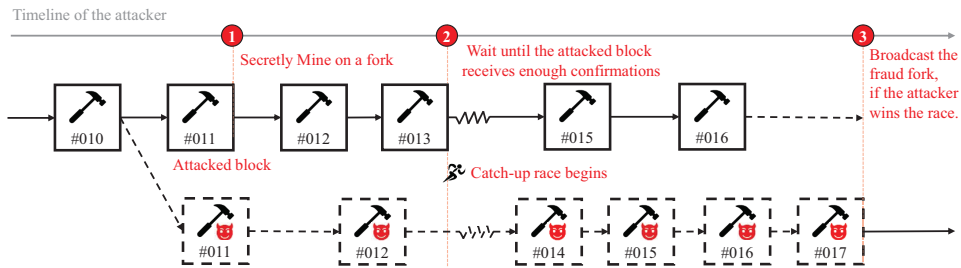}

\end{raggedright} \centering{}\caption{The attacker's scheme in a double spending attack.\label{fig:workflowDSA}}
\vspace{-0.6cm}
\end{figure}

Recall that the benign miners' and an attacker's mining rate are $\lambda^{b}$
and $\beta\lambda^{b}$, respectively, as $\beta$ represents the
attacker's relative hash power.  Rational benign miners always work
on the longest chain for the mining rewards, while the attacker attempts
to generate a fraudulent fork to revise the history for double spending
or other purposes. Before initiating a fraud, the attacker will broadcast
a regular event (e.g., paying for access) and wait for it to be assembled
into a block. That block is the recorded history which the attacker
attempts to revise. As shown in Fig. \ref{fig:workflowDSA}, the
attacker will mount an alternative history attack according to the
following steps:
\begin{enumerate}
\item secretly mine on a fork excluding the attacked block;
\item wait until the attacked block receives enough confirmations (say,
$N$ blocks);
\item broadcast the respective blocks as soon as the fraudulent fork is
longer than the benign fork.
\end{enumerate}
Note that the attacker could always mine on the fraudulent chain,
despiting being many blocks behind. In practice, the attacker usually
will stop working on the fraudulent chain and give up if it is, e.g.,
$N_{g}$ blocks, behind.

Given the attacker's relative mining rate $\beta$ and its strategy
$N_{g}$, we would like to determine the probability of a successful
alternative history attack, donated by $\mathsf{S}\left\{ N,\beta,N_{g}\right\} $,
to reflect how secure B-RAN is.
\begin{thm}
\label{thm:probattack}Given the attacker's relative mining rate $\beta$
and the give-up strategy $N_{g}$, the probability of a successful
alternative history attack is
\begin{equation}
\mathsf{S}\left(N,\beta,N_{g}\right)=\begin{cases}
1-\sum_{n=0}^{N}\binom{n+N-1}{n}\left(\frac{1}{1+\beta}\right)^{N}\left(\frac{\beta}{1+\beta}\right)^{n}\left(\frac{1-\beta^{N-n+1}}{1-\beta^{N_{g}+1}}\right) & \text{\text{if }}\beta\neq1\\
1-\sum_{n=0}^{N}\frac{1}{2^{N+n}}\binom{n+N-1}{n}\left(\frac{N-n+1}{N_{g}+1}\right) & \text{\text{if }}\beta=1.
\end{cases}\label{security:probfinity}
\end{equation}
\end{thm}
\begin{proof}
Please see Appendix A.
\end{proof}
A more powerful attacker with a larger $\beta$ is more likely to
revise history successfully, while more confirmations can effectively
mitigate the risk of such malicious attacks. However, no matter how
large $N$ is, the alternative history attack is always possible if
the attacker has a non-negative mining rate and sufficient luck. We
note that the attacker's give-up parameter $N_{g}$ has never been
characterized in the existing studies, which, in fact, cannot be simply
ignored. Based on Theorem \ref{thm:probattack}, we further consider
the impact of $N_{g}$ on the success probability.
\begin{cor}
\label{cor:infinity}Given $N$ and $\beta$, the probability of a
successful attack $\mathsf{S}\left(N,\beta,N_{g}\right)$ increases
monotonically with $N_{g}$. Hence, the maximum success probability
is $\mathsf{S}\left(N,\beta\right)=\lim_{N_{g}\rightarrow\infty}\mathsf{S}\left(N,\beta,N_{g}\right)$,
given by
\begin{equation}
\mathsf{S}\left(N,\beta\right)=\begin{cases}
1-\sum_{n=0}^{N}\binom{n+N-1}{n}\left(\frac{1}{1+\beta}\right)^{N}\left(\frac{\beta}{1+\beta}\right)^{n}\left(1-\beta^{N-n+1}\right) & \text{\text{if }}\beta<1\\
1 & \text{\text{if }}\beta\ge1.
\end{cases}\label{Security:probinfinity}
\end{equation}
\end{cor}
\begin{proof}
Please see Appendix B.
\end{proof}
 Corollary \ref{cor:infinity} intuitively indicates that the attacker
can increase the success probability through repeated trials. Hence,
$\mathsf{S}\left(N,\beta\right)$ is the maximum probability of an
alternative history attack. The dominant strategy of a powerful attacker
with $\beta\geq1$ is to never quit. Because $\mathsf{S}\left(N,\beta\right)=1$
for $\beta>1$, the fraud would succeed eventually. This is known
as the ``51\%'' attack or Goldfinger attack\cite{Tschorsch2016}.
For weaker hash power $\beta<1$, although a larger $N_{g}$ can also
increase the success probability, it costs a vast number of resources
and could possibly lead to negative yields. If the fraudulent chain
is hundreds of blocks behind the benign chain, a rational attacker
should give up and mount a new round of attack instead of sticking
to the original one with mounting cost.

\begin{figure}
\centering %
\begin{minipage}[t]{0.47\linewidth}%
\begin{raggedright} \centering\includegraphics[width=0.99\textwidth]{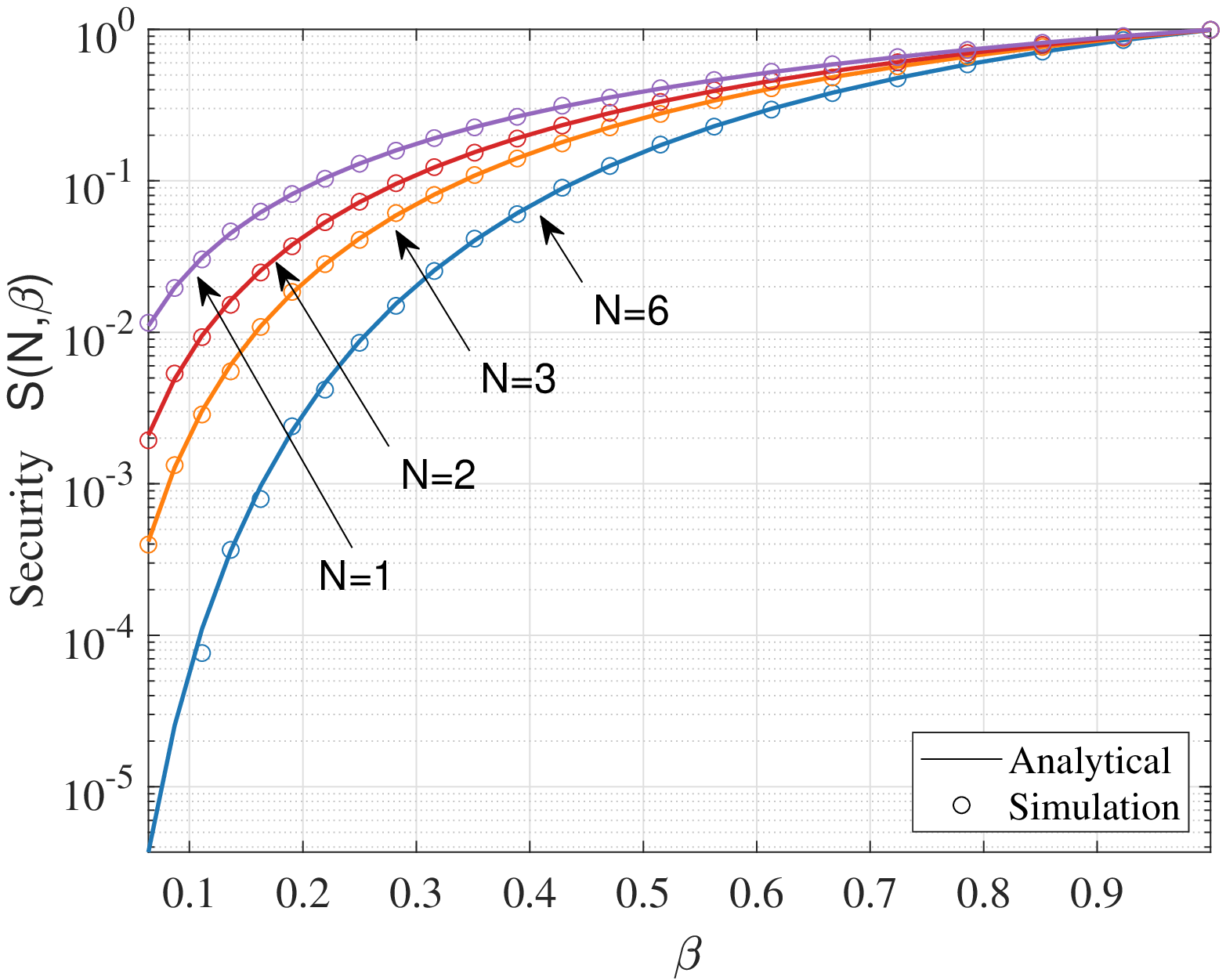}

\end{raggedright} \centering{}\caption{The analytical and experimental probability of success attack for
different $\beta$ and $N$ with $N_{g}=\infty$.\label{fig:security1}}
\end{minipage}\hfill{}%
\begin{minipage}[t]{0.47\linewidth}%
\begin{raggedright} \centering\includegraphics[width=0.99\textwidth]{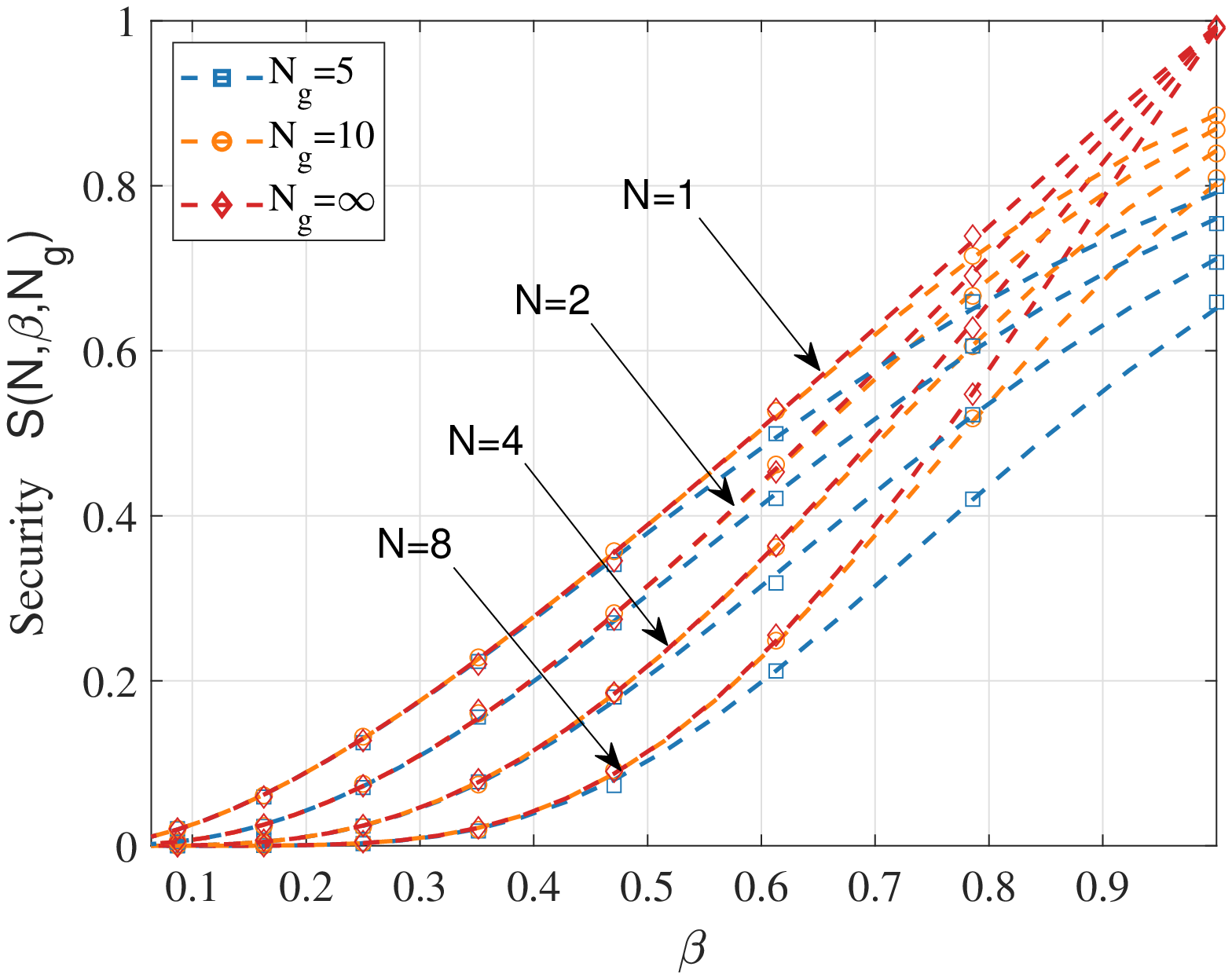}

\end{raggedright} \centering{}\caption{The impact of the attacker's strategy $N_{g}$ on the probability
of an alternative history attack.\label{fig:security2}}
\end{minipage}\vspace{-0.6cm}
\end{figure}
We would like to point out why our result differs from the existing
studies. The original Bitcoin paper\cite{Nakamoto2008} claimed that
$Y$ follows a Poisson distribution, which is actually inaccurate.
In \cite{Tschorsch2016}, authors counted the breakeven case in computing
the probability of successful attack. Authors of \cite{Tschorsch2016,Rosenfeld2014}
assumed that the attacker can pre-mine a block before mount an attack.
 Clearly, the existing works relied on less realistic assumptions.
Our results further reveal the impact of the attacker's strategy ($N_{g}$),
and therefore is more practical compared to the existing studies.

In \textbf{Prototype Verification D}, we now show the security performance
of B-RAN obtained from our home-built prototype. From Fig. \ref{fig:security1},
one can see clearly that the experiment results firmly agree with
our analytical probability of Theorem \ref{thm:probattack}. As the
relative mining rate $\beta$ increases, the probability of a successful
attack rises significantly. The chain is more secure if it is safeguarded
by more confirmations, which suggests that the number of confirmations
$N$ can effectively reduce the risk.

In Fig. \ref{fig:security2}, we illustrate that the attacker's threshold
$N_{g}$ can also affect the success probability, especially for $\beta$
close to $1$. An attacker can increase the success probability by
continued attempts, which, although, expends more mining energy. For
$\beta>1$, the 51\% attack would eventually succeed, and thus the
dominant strategy for a powerful attacker should be to use $N_{g}=\infty$.
According to Fig. \ref{fig:security2}, $\mathsf{S}\left(N,\beta\right)=\mathsf{S}\left(N,\beta,N_{g}=\infty\right)$
serves as an upper bound of the success probability, which corroborates
Corollary \ref{cor:infinity}. In practice, we can use the upper bound
$\mathsf{S}\left(N,\beta\right)$ to estimate the security level of
B-RAN.

\section{Latency-Security Trade-off}

From the above analysis, we discover that there exists an inherent
relation between latency and security. This latency-security relationship
and design trade-off capture a more complete picture on achievable
performance of B-RAN. We formally state this design trade-off as the
following theorem.
\begin{thm}
\label{thm:trade-off}The trade-off between latency and security is
given by the piecewise-linear function connecting the points $\left(\mathsf{L}\left(N,\Phi\right),\mathsf{S}\left(N,\beta\right)\right)$
for $N=1,2,...$, where $\mathsf{L}\left(N,\Phi\right)$ and $\mathsf{S}\left(N,\beta\right)$
is given by \eqref{latency:expression} and \eqref{Security:probinfinity},
respectively.
\end{thm}
Given the network setting $\Phi$ and the attacker's relative hash
power $\beta$, the performance of B-RAN can be thoroughly evaluated
by the trade-off specified Theorem \ref{thm:trade-off}. The latency-security
trade-off can be used as a new performance metric to characterize
B-RAN. Theorem \ref{thm:trade-off} indicates that, more confirmations
$N$ can effectively safeguard the security of B-RAN $\mathsf{S}\left(N,\beta\right)$,
but would unfortunately increase the access delay $\mathsf{L}\left(N,\Phi\right)$.
Conversely, fewer confirmations to verify a request can considerably
reduce latency while suffering higher risk from an adversary.

\begin{figure}
\begin{raggedright} \centering\subfigure[]{\includegraphics[width=0.45\textwidth]{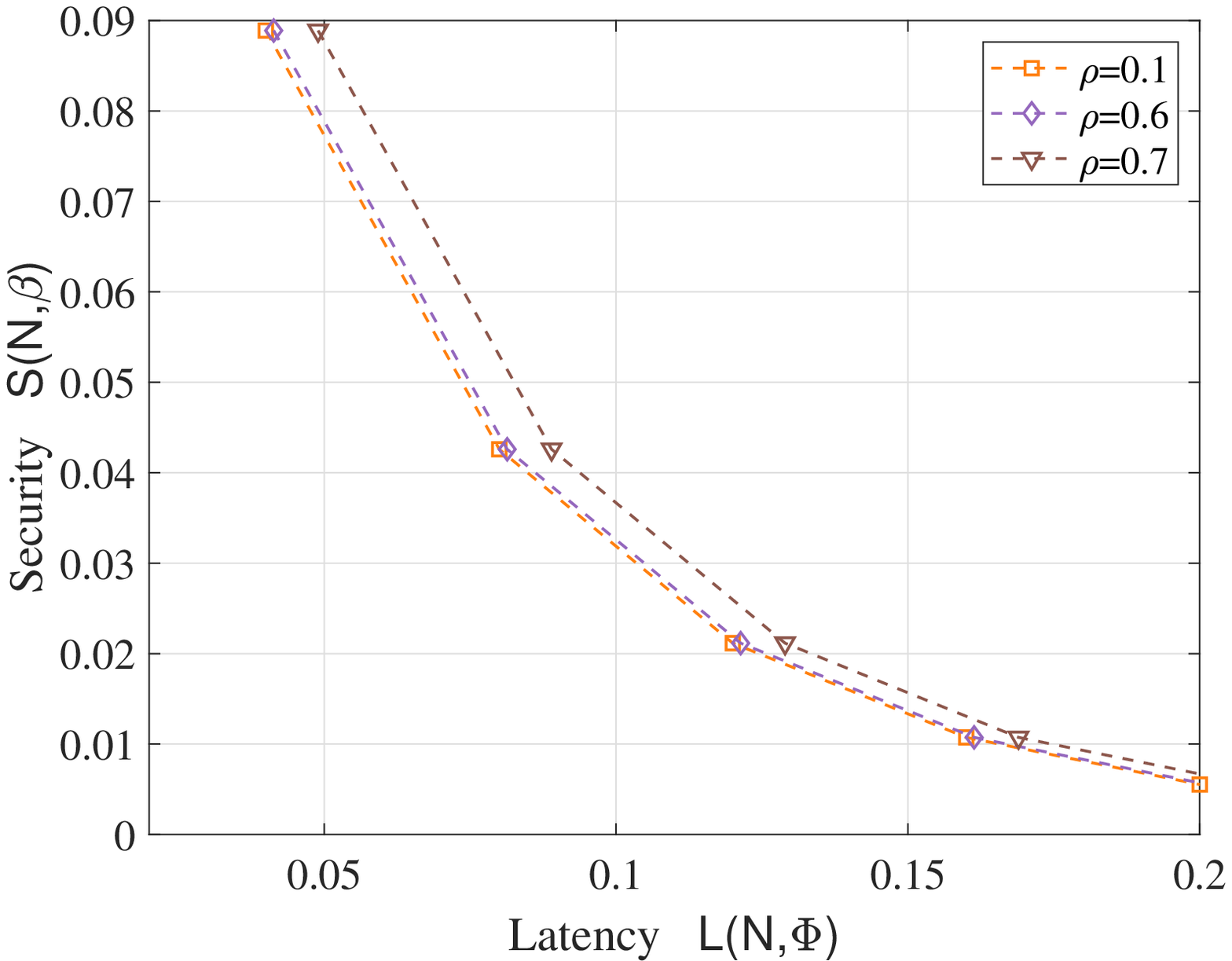}}\hfill{}
\subfigure[]{\includegraphics[width=0.45\textwidth]{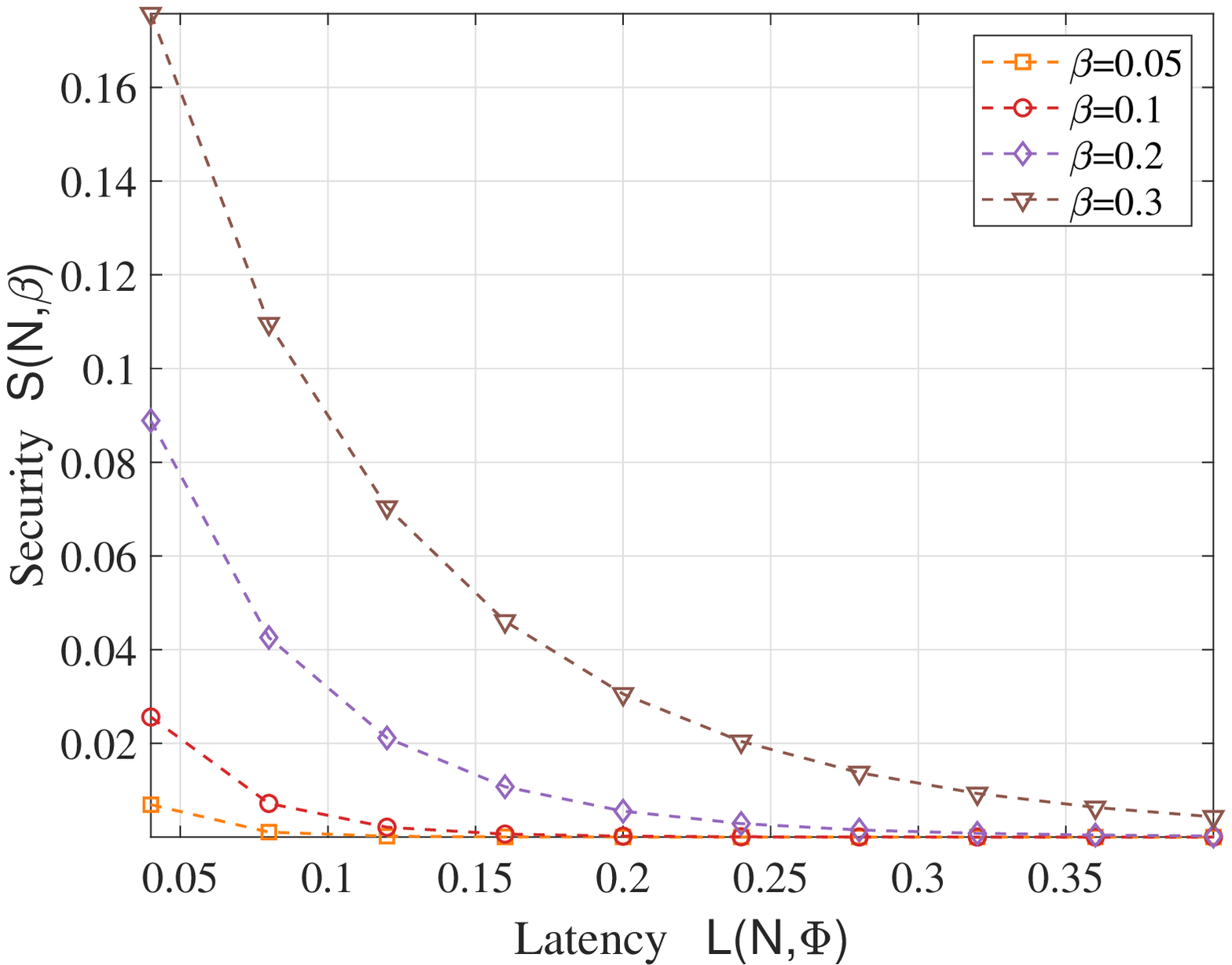}}

\end{raggedright} \centering{}\caption{The latency-security trade-off ($s=25$ and $T^{b}=0.04$). (a) Under
different traffic intensity $\rho$ and $\beta=0.2$. (b) An attacker
with different mining rate $\beta$ and $\rho=0.1$. \label{fig:tradeoffcurve}}
\vspace{-0.6cm}
\end{figure}

Fig. $\ref{fig:tradeoffcurve}$ graphically illustrates the trade-off
curve $\left(\mathsf{L}\left(N,\Phi\right),\mathsf{S}\left(N,\beta\right)\right)$.
Essentially, the latency-security trade-off is between the convergence
rate and confirmation error probability of a distributed system. Enhancing
the security advantage comes at a price of longer access delay, whereas
latency reduction comes at a price of security. Theorem \ref{thm:trade-off}
provides the achievable region of B-RAN for given network parameters
$\Phi$ and the attacker's power $\beta$. The trade-off provides
a more in-depth and complete view of B-RAN than just staring at latency
or security.

Theorem \ref{thm:trade-off} also makes another informative statement.
If the network condition $\Phi$ becomes better, e.g., with less traffic
load $\lambda^{a}$, then the entire trade-off curve will be shifted
to the left along the axis of latency for the same security level.
If the attacker becomes more powerful, e.g., a larger $\beta$, then
the trade-off curve drops in the term of security at the same access
delay. Nevertheless, the latency bounds in Section V-B is still application,
and can provide an estimation for the achievable performance of B-RAN.

Furthermore, Theorem \ref{thm:trade-off} implies that latency and
security can be balanced by adjusting the number of confirmations
$N$. The confirmation number $N$ is a positive integer that can
be chosen to achieve the shortest latency under an acceptable security
condition. Different from cryptocurrencies, wireless access is much
more sensitive to latency. Therefore, we usually should set a relative
smaller $N$ instead of the typical six confirmations suggested in
Bitcoin\cite{Nakamoto2008}. The specific value of $N$ should be
determined by the latency requirement and also the security level.

\begin{figure}
\begin{raggedright} \centering\includegraphics[width=0.5\textwidth]{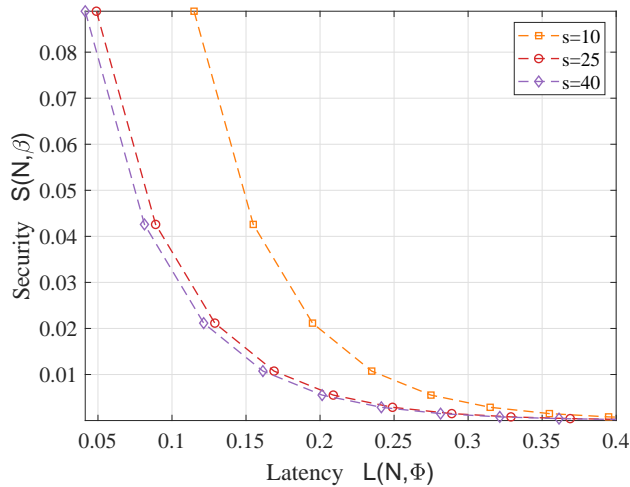}\end{raggedright}
\centering{}\caption{The impact of the network size on the trade-off curves. ($T^{b}=0.04$,
$\rho=0.7$ and $\beta=0.2$). \label{fig:tradeoffcurve-s}}

\vspace{-0.6cm}
\end{figure}

It is worth noting that the network size can influence the trade-off.
Fig. \ref{fig:tradeoffcurve-s} demonstrates that, a more expansive
network with more access links $s$ can significantly improve the
network performance. Given the security level $\beta$ and traffic
intensity $\rho$, B-RAN with a larger $s$ has shorter service latency
and allows more confirmations. Fig. \ref{fig:tradeoffcurve-s} indicates
that, as B-RAN unites more subnetworks, it means more access links
$s$ and higher value for B-RAN. This effect corresponds to the well-known
positive network effect. As more subnetworks join B-RAN, they make
the B-RAN more valuable and a longer bandwagon.

\section{Conclusions}

In this study, we established an original framework to analytically
characterize B-RAN properties and performances. Based on the established
queuing model, we analyzed the system latency of B-RAN and evaluated
the security level of B-RAN by considering the risk of alternative
history attacks. We uncovered an inherent latency and security relationship.
We proposed to assess the achievable performance of B-RAN according
to the latency-security trade-off curve. We validated our analytical
works by building an in-house B-RAN prototype. The results derived
from the model will provide meaningful insights and guidelines for
future B-RAN designs and implementations.

\appendix{}

\subsection{Proof of Theorem \ref{thm:probattack}}
\begin{proof}
Assume that the mining rates and the difficulty remain constant. The
probability of extending honest chain by one block is $\frac{1}{1+\beta}$,
and the probability for an attacker to find the next block is $\frac{\beta}{1+\beta}$.
Hence, the mining process can be described as a sequence of independent
Bernoulli trials with the probability of success $\frac{1}{1+\beta}$.
The attacker has to deliberately wait for $N$ confirmations, i.e.,
to let the benign chain continue to grow by $N$ additional blocks.
Meanwhile, the attacker can secretly produce $Y$ block in the fraudulent
fork. We count the number of failures until $N$ successes. As a result,
the random number of failures $Y$ obeys a negative binomial distribution
$Y\sim{\cal NB}(N,1/(1+\beta))$ with the probability mass function
\[
\text{Pr}\left\{ Y=n;N,\frac{1}{1+\beta}\right\} =\binom{n+N-1}{n}\left(\frac{1}{1+\beta}\right)^{N}\left(\frac{\beta}{1+\beta}\right)^{n}.
\]

Once $N$ blocks are found by the honest network, in a period of time
during which $Y=n$ blocks are found by the attacker. A race between
honest miners and the attacker begins.  As soon as the fraudulent
chain becomes longer than the benign chain, the attacker can publish
the fraudulent chain to alter a confirmed history. But, if the fraudulent
chain is $N_{g}$ (typically greater than $N$) blocks behind the
benign chain, the attacker would give up. Let $P_{n}=\text{Pr}\left\{ \text{Win}|z=n\right\} $
be the probability that the attacker wins before giving up, when starting
from $n$ blocks behind. Obviously, we have $P_{-1}=1$ and $P_{N_{g}}=0$.
The derivation of this recursion is simple. If the attacker finds
the next block, the fraudulent chain is $n-1$ blocks shorter than
the benign chain and the probability of a successful attack becomes
$P_{n-1}$. Similarly, if the benign miners find a block, the attacker
is $n+1$ blocks behind and the probability of a successful attack
becomes $P_{n+1}.$ Applying the condition on the outcome of the first
generated block, we have
\begin{align}
P_{n} & =\frac{1}{1+\beta}P_{n+1}+\frac{\beta}{1+\beta}P_{n-1},\quad0\le n<N_{g}.\label{eq:recursionEquation}
\end{align}
Note that \eqref{eq:recursionEquation} comes differently from the
equilibrium equations in Section IV. We further rewrite \eqref{eq:recursionEquation}
as
\[
P_{n-1}-P_{n}=\frac{1}{\beta}\left(P_{n}-P_{n+1}\right),\quad0\le n<N_{g}.
\]
For $n=N_{g}-1$, we have
\[
P_{N_{g}-2}-P_{N_{g}-1}=\frac{1}{\beta}\left(P_{N_{g}-1}-P_{N_{g}}\right)=\frac{1}{\beta}P_{N_{g}-1},
\]
which, via recursive, yields

\[
P_{N_{g}-n-1}-P_{N_{g}-n}=\frac{1}{\beta^{n}}P_{N_{g}-1},\quad0\le n<N_{g}.
\]
Hence,
\begin{align*}
P_{N_{g}-n-1} & =P_{N_{g}-1}+\sum_{m=1}^{n}\frac{1}{\beta^{m}}P_{N_{g}-1}=\begin{cases}
P_{N_{g}-1}\frac{1-1/\beta^{n+1}}{1-1/\beta}, & \text{if}\;\beta\neq1\\
P_{N_{g}-1}\left(n+1\right), & \text{if}\;\beta=1.
\end{cases}
\end{align*}
Using the boundary condition that $P_{-1}=1$ gives
\begin{equation}
P_{N_{g}-1}=\begin{cases}
\frac{1-1/\beta}{1-1/\beta^{N_{g}+1}} & \text{if}\;\beta\neq1\\
\frac{1}{N_{g}+1} & \text{if}\;\beta=1.
\end{cases}\label{eq:limit-1}
\end{equation}
We thus obtain the expression of $P_{n}$ as:
\begin{equation}
P_{n}=\begin{cases}
\frac{\beta^{n+1}-\beta^{N_{g}+1}}{1-\beta^{N_{g}+1}} & \text{if}\;\beta\neq1\text{ and }0\le n<N_{g}\\
\frac{N_{g}-n}{N_{g}+1} & \text{if}\;\beta=1\text{ and }0\le n<N_{g}\\
1 & \text{if}\;n<0\\
0 & \text{if}\;n\geq N_{g}.
\end{cases}\label{security:Pi}
\end{equation}

Now assume that the benign chain extends $N$ blocks and the fraudulent
chain extends $n$ blocks. The attacker begins the race with $N-n$
blocks behind. The probability of a successful alternative history
attack can be expressed as
\begin{align}
\mathsf{S}\left(N,\beta,N_{g}\right) & =\sum_{n=0}^{\infty}\text{Pr}\left\{ \text{\text{Win}}|z=N-n\right\} \text{Pr}\left\{ Y=n;N,\frac{1}{1+\beta}\right\} \nonumber \\
 & =\sum_{n=0}^{\infty}\binom{n+N-1}{n}\left(\frac{1}{1+\beta}\right)^{N}\left(\frac{\beta}{1+\beta}\right)^{n}P_{N-n}.\label{security:tmp}
\end{align}
We can arrange the terms in \eqref{security:tmp} by using the fact
$\sum_{n=0}^{\infty}\binom{n+N-1}{n}\left(\frac{1}{1+\beta}\right)^{N}\left(\frac{\beta}{1+\beta}\right)^{n}=1$,
and finally obtain the result \eqref{security:probfinity} in Theorem
\ref{thm:probattack}.
\end{proof}

\subsection{Proof of Corollary \ref{cor:infinity}}
\begin{proof}
The monotonicity of $N_{g}$ can be easily obtained by observing the
expression of $\mathsf{S}\left(N,\beta,N_{g}\right)$. As $N_{g}\rightarrow\infty$,
\eqref{security:Pi} in Appendix A becomes
\[
\lim_{N_{g}\rightarrow\infty}P_{n}=\begin{cases}
\beta^{n+1} & \text{\text{if }}\beta<1\text{ and }i>0\\
1 & \text{otherwise}.
\end{cases}
\]
As a direct result, we easily obtain $\mathsf{S}\left(N,\beta\right)=\lim_{N_{g}\rightarrow\infty}\mathsf{S}\left(N,\beta,N_{g}\right)$
in the limiting case from \eqref{security:probfinity}.
\end{proof}
\bibliographystyle{IEEEtran}
\bibliography{IEEEabrv,blockchain}

\end{document}